\documentclass[a4paper,twocolumn,notitlepage,nofootinbib,longbibliography,superscriptaddress,floatfix]{revtex4-2}

\usepackage{
    adjustbox,
    algorithm,
    algpseudocode,
    amsmath,
    amssymb,
    amsthm,
    booktabs,
    braket,
    csquotes,
    enumitem,
    graphicx,
    IEEEtrantools,
    mathtools,
    multirow,
    pifont,
    placeins,
    refcount,
    relsize,
    tabls,
    tcolorbox,
    tikz-cd,
    times,
    svg,
    xcolor
}

\definecolor{mainblue}{HTML}{1f77b4}
\definecolor{mainorange}{HTML}{ff7f0e}
\definecolor{maingreen}{HTML}{2cc92c}
\definecolor{mainred}{HTML}{DC3522}
\definecolor{mainpurple}{HTML}{7b5ec6}
\definecolor{mainpink}{HTML}{ff99dd}
\definecolor{pantoneblue}{HTML}{19c2a9}

\usepackage[colorlinks=true,citecolor=mainpurple,linkcolor=pantoneblue,urlcolor=mainpink]{hyperref}

\newcommand{\ie}{\textit{i.e.}\ }
\newcommand{\eg}{\textit{e.g.}\ }

\newtheorem{proposition}{Proposition}
\newtheorem{definition}{Definition}

\newtheorem{property}{Property}
\newtheorem{observation}{Observation}
\newtheorem{class}{Class}

\DeclareMathOperator{\poly}{poly}

\DeclareMathOperator{\Tr}{tr}

\DeclareMathOperator{\BQP}{\mathsf{BQP}}
\DeclareMathOperator{\BPP}{\mathsf{BPP}}
\DeclareMathOperator{\Heurpoly}{\mathsf{HeurBPP/poly}}
\DeclareMathOperator{\BPPqgenpoly}{\mathsf{BPP/qgenpoly}}
\DeclareMathOperator{\Ppoly}{\mathsf{P/poly}}

\providecommand{\hR}{\ensuremath{\hat{R}}}

\providecommand{\to}{\ensuremath{\Tilde{o}}}

\providecommand{\calA}{\ensuremath{\mathcal{A}}}

\providecommand{\calD}{\ensuremath{\mathcal{D}}}

\providecommand{\calF}{\ensuremath{\mathcal{F}}}

\providecommand{\calO}{\ensuremath{\mathcal{O}}}
\providecommand{\calP}{\ensuremath{\mathcal{P}}}

\providecommand{\calX}{\ensuremath{\mathcal{X}}}
\providecommand{\calY}{\ensuremath{\mathcal{Y}}}

\providecommand{\bbI}{\ensuremath{\mathbb{I}}}

\newif\ifverbose
\newcommand{\EGF}[1]{\ifverbose\textcolor{brown}{[EGF: #1]}\fi}

\verbosefalse
\newcommand{\PsiMPS}{\ensuremath{\Psi_{\operatorname{MPS}}}}

\newcommand{\fu}{Dahlem Center for Complex Quantum Systems, Freie Universit\"{a}t Berlin, 14195 Berlin, Germany}
\newcommand{\hzb}{Helmholtz-Zentrum Berlin f{\"u}r Materialien und Energie, 14109 Berlin, Germany}
\newcommand{\hhi}{Fraunhofer Heinrich Hertz Institute, 10587 Berlin, Germany}
\newcommand{\bsc}{Barcelona Supercomputing Center, Barcelona 08034, Spain}
\newcommand{\ub}{Universitat de Barcelona, Barcelona 08007, Spain}

\begin{document}

\title{Prospects for quantum advantage in machine learning from the representability of functions}

\author{Sergi Masot-Llima}
\thanks{These authors contributed equally.}
\affiliation{\ub}
\affiliation{\bsc}
\author{Elies Gil-Fuster}
\thanks{These authors contributed equally.}
\affiliation{\fu}
\affiliation{\hhi}
\author{Carlos Bravo-Prieto}
\affiliation{\fu}
\author{Jens Eisert}
\affiliation{\fu}
\affiliation{\hhi}
\affiliation{\hzb}
\author{Tommaso Guaita}\email[Corresponding author: ]{tommaso.guaita@fu-berlin.de}
\affiliation{\fu}

\begin{abstract}
Demonstrating quantum advantage in machine learning tasks requires navigating a complex landscape of proposed models and algorithms. To bring clarity to this search, we introduce a framework that connects the structure of parametrized quantum circuits to the mathematical nature of the functions they can actually learn. Within this framework, we show how fundamental properties, like circuit depth and non-Clifford gate count, directly determine whether a model's output leads to efficient classical simulation or surrogation. We argue that this analysis uncovers common pathways to dequantization that underlie many existing simulation methods. More importantly, it reveals critical distinctions between models that are fully simulatable, those whose function space is classically tractable, and those that remain robustly quantum. This perspective provides a conceptual map of this landscape, clarifying how different models relate to classical simulability and pointing to where opportunities for quantum advantage may lie.

\end{abstract}

\maketitle

\section{Introduction}\label{s:intro}
	
    \emph{Quantum machine learning} (QML) is recognized as a promising approach to harness quantum computing for learning tasks~\cite{dunjko2018machine,cerezo2021variational,QMLWhitePaper}.
    As with all quantum algorithms, a central question is whether QML holds potential for quantum advantage~\cite{arute2019quantum,schuld2022quantum,MindTheGaps,GrandChallenge} over classical computing.
    The counter-narrative to quantum advantage is \emph{dequantization}, where upon close inspection certain quantum algorithms yield no benefit over classical counterparts, as one can classically solve the task at hand. 
    Dequantization of quantum algorithms for \emph{machine learning}, in particular, has seen a surge of interest in recent years, leaving few claims of quantum advantage unchallenged~\cite{schreiber2023classical, cerezo2023does, landman2022classicallyapproximatingvariationalquantum, sweke2025potential, gilfuster2025relation}.
    
    While QML models for classical data can be studied from several perspectives, significant theoretical developments have emerged from investigating the function families that \emph{parametrized quantum circuits} (PQCs) can give rise to~\cite{schuldEffectDataEncoding2021, huang2021power, schuld2021kernels, jerbi2023quantum, schreiber2023classical, landman2022classicallyapproximatingvariationalquantum}.
    Characterizing the functional forms arising from PQCs allows us to delineate the boundaries of quantum learning and guide the search for advantage.
    Specifically, the concept of \emph{representability} of functions can help establish links between, e.g.,  complexity theory and the structural properties of PQCs.
    Being a close relative of \emph{classical simulation} of quantum circuits, representability is a 
    weaker notion that enables a broader discussion on dequantization and \emph{classical surrogates}.
    This is the specific perspective we explore in this work, as depicted in Fig.~\ref{fig:concept}.
    
    We begin with reviewing established classical techniques that can be used to classically simulate PQCs, such as tensor networks~\cite{berezutskii2025tensornetworksquantumcomputing,vidal2003EffClassSim}, 
    stabilizer methods \cite{aaronsonImprovedSimulationStabilizer2004, bravyiSimulationQuantumCircuits2019c}, Pauli back-propagation~\cite{schuster2024polynomialtimeclassicalalgorithmnoisy,rudolph2023classical}, and Lie algebraic methods~\cite{goh2025liealgebraicclassicalsimulationsquantum}.
    Alongside this, we revisit recent efforts to characterize the output functions of these PQCs~\cite{shinDequantizingQuantumMachine2023,cristoiu2024}.
    Building on these foundations, we propose a unified framework that illustrates \emph{how} all these methods relate to each other, based solely on the function classes they implicitly target.
    
    To structure this framework, we mirror the foundational work of Ref.~\cite{gyurik2023exponential} by identifying two key properties that capture the landscape of function families: \emph{evaluation} and \emph{identification}.
    Crucially, PQCs with restricted resources may produce function families that are efficiently \emph{evaluatable} on a classical computer, yet existing simulation techniques may fail to \emph{identify} the specific function corresponding to a given set of parametrized gates.
    \EGF{I rewrote the end of this paragraph}
    In the context of quantum advantage in learning, the distinction between evaluation and identification underpins several milestone results, including the characterization of exponential separations~\cite{gyurik2023exponential}, the potential for data to overcome identification hardness~\cite{huang2021power}, and using classical surrogates as learning models~\cite{schreiber2023classical}.
    
    Our goal is to highlight the usefulness of analyzing QML and quantum advantage \emph{through the lens of function families}, specifically via succinct classical representations of quantum functions.
    Our main contribution is a framework that thoroughly organizes prior art in a natural language.
    Central to this framework is the association of a function basis with a PQC architecture, allowing us to specify whether the resulting functions are high- or low-rank on this basis.
    \EGF{The following sentence is new.}
    While we do not claim technical novelty regarding the underlying methods, our framework distinguishes itself by focusing on the functions \emph{produced by a QML model} (the hypothesis family), rather than the functions intended to be learned (the concept class or ground truth).
    
    One family of circuits that stands out in our framework is a special case of the so-called \emph{flipped models}~\cite{Jerbi2024shadows}, where the measurement has limited quantum resources, which we study to showcase our approach.
    Furthermore, we reformulate established results within this language, which clearly spell out the worst-case hardness of the optimization tasks that arise in supervised learning.
    While we do not claim immediate breakthroughs toward practically-relevant quantum advantage in learning, we believe the framework proposed in this perspective will accelerate progress in the field.

    \begin{figure*}[t]
    	\centering
    	\includegraphics[width=0.9\linewidth]{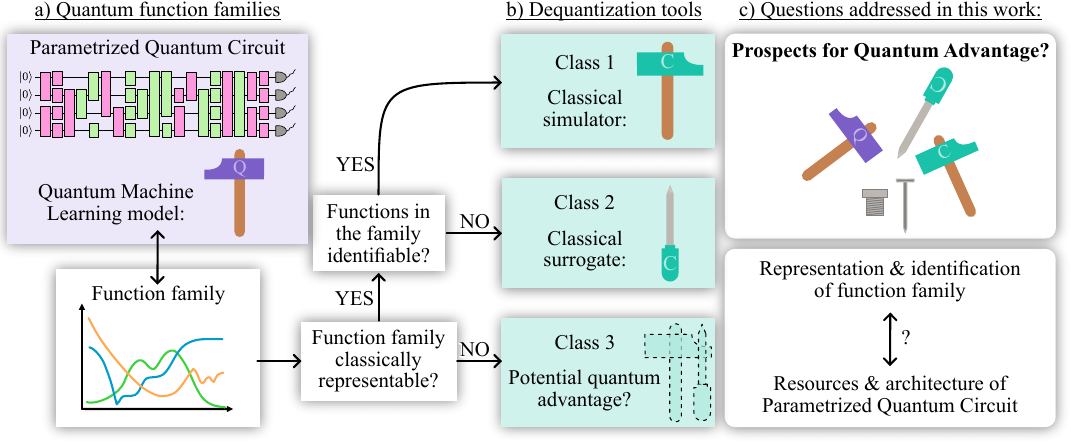}
    	\caption{Overview of the proposed framework. (a) We establish a mapping between the structure of PQCs and the hypothesis families $f$ they generate. (b) Based on the classical evaluatability of $f$ and its identifiability from the circuit parameters, we categorize PQCs into three distinct classes. (c) This classification provides a perspective for addressing key open questions in the field, such as the potential for quantum advantage and the limits of efficient classical simulation.}
    	\label{fig:concept}
    \end{figure*}

\section{Framework}\label{s:preliminaries}

    To provide a rigorous basis for our perspective, we first establish formal definitions regarding supervised learning and the dequantization of learning models. These definitions underpin the classification of PQCs presented in the subsequent sections.
    
\subsection{Supervised learning}\label{ss:supervised}

    Let $\calX$ denote the input (data) domain and $\calY$ the output (label) co-domain.
    Let further $\calD$ be an unknown probability distributionover $\calX\times\calY$, which produces input-output pairs $(x,y)\sim\calD$ according to some pattern that must be learned.
    For any function $f\colon\calX\to\calY$, the \emph{expected risk} functional $R_\calD$ quantifies the quality of $f$ in capturing the pattern hidden in $\calD$.
    In the simplest case, $R_\calD$ measures the discrepancy between $f(x)$ and $y$ averaged over the distribution.
    The pair $(\calD,R_\calD)$ thus constitutes a \emph{learning task}.

    In \emph{supervised learning}, the learner is provided with a \emph{training set} $S=\{(x_i,y_i)\}_{i=1}^N$ drawn i.i.d. from $\calD$, from which it attempts to infer the underlying pattern. 
    This process employs a \emph{learning model}, defined as a pair $(\calF, \calA)$, where $\calF$ is a set of functions termed the \emph{hypothesis family} and $\calA$ is a \emph{training algorithm}.
    Given a training set $S$, the algorithm $\calA$ returns a hypothesis $\calA(S) \in \calF$.

    In \emph{quantum machine learning} (QML), we focus primarily on models based on \emph{parametrized quantum circuits} (PQCs).
    These circuits are characterized by an initial quantum state, a sequence of quantum gates specified by classical parameters, and a measurement observable.
    Here, the hypothesis family consists of functions that can be 
    evaluated using a specific PQC architecture.
    We distinguish two main categories of QML models: \emph{variational} QML models, where the hypothesis family is explicitly \emph{parametrized} by the circuit parameters; and \emph{quantum kernel methods}, where PQC evaluations are combined via classical post-processing.
    In both approaches, the training subroutine is assumed to be \emph{classically efficient}, using black-box access to the hypothesis family and its derivatives.
    For conceptual clarity, we exclude quantum kernel methods from the following analysis unless explicitly mentioned.
    We provide more details on variational QML models in Definition~\ref{def:param_model} in Section~\ref{ss:qml}.

    The goal of the training algorithm is to identify a hypothesis that minimizes the expected risk $R_\calD$. However, since $\calD$ is unknown, the expected risk is not accessible. In practice, we optimize a proxy functional called the \emph{empirical risk} $\hR_S$, which quantifies the fit of a function $f$ to the specific training set $S$.
	The optimization problem of finding the hypothesis in $\calF$ that minimizes $\hR_S$ is known as \emph{empirical risk minimization} (ERM). Note that ERM is generally computationally hard~\cite{bittel2021training}. 

\subsection{A taste of dequantization in learning}\label{ss:dequantization}

    Following Ref.~\cite{gilfuster2025relation}, we define \emph{quantum advantage} and \emph{dequantization} relative to a specific learning task $(\calD,R_\calD)$.
    We say that an efficient QML model $(\calF,\calA)$ exhibits quantum advantage if it yields significantly better performance on the task compared to any efficient classical learning model.
    Formally, let $(\calF_C,\calA_C)$ be the best efficient classical learning model for the task. Quantum advantage requires that, with high probability over training sets $S\sim\calD^N$, there exists a significant (ideally superpolynomial) gap in the expected risks of the resulting solutions: $R_\calD(\calA(S))\ll R_\calD(\calA_C(S))$.
    Conversely, we say the QML model admits \emph{dequantization} if a classical model achieves similar or better performance, satisfying $R_\calD(\calA(S))\gtrsim R_\calD(\calA_C(S))$.

    Intuitively, for dequantization to be possible, the set of optimal solutions must be accessible via classically efficient functions.
    We identify two distinct properties that characterize this efficiency for a parametrized hypothesis family $\calF=\{f_\vartheta\,|\,\vartheta\in\Theta\}$,  where $\Theta$ denotes a set of classical variational parameters. 
    
    \begin{property}[Efficient evaluation]\label{pty:evaluatable}
		The functions $f_\vartheta(x)\in\calF$ are \emph{classically efficiently evaluatable} if, for every choice of $\vartheta$, there exists a polynomial-time classical algorithm that computes the map $x\mapsto f_\vartheta(x)$ for any input $x\in\calX$. 
	\end{property}

    This definition is deeply linked to the notion of \emph{non-uniform complexity} and \emph{advice}: the specific classical algorithm that evaluates each function may be very different for different choices of $\vartheta$.
    As discussed below, efficient evaluation alone may be insufficient for dequantization.
    For instance, even if a polynomial-time classical algorithm exists for a specific $f_\vartheta$, determining that algorithm from the parameters $\vartheta$ may be computationally intractable.
    We, therefore, define a stricter property to preclude this difficulty:
    \begin{property}[Efficient identification]\label{pty:identifiable}
		The functions $f_\vartheta(x)\in\calF$ are \emph{classically efficiently identifiable} if there exists a polynomial-time classical algorithm that computes the joint map $(x,\vartheta)\mapsto f_\vartheta(x)$ for any $x\in\calX$ and $\vartheta\in\Theta$.
	\end{property}

    In the usual language of quantum computing, Property~\ref{pty:identifiable} corresponds to the exact \emph{classical simulation} of the quantum circuit.
    Property~\ref{pty:evaluatable}, by contrast, is the focus of recent literature on \emph{classical surrogates}~\cite{schreiber2023classical, landman2022classicallyapproximatingvariationalquantum, Jerbi2024shadows, sweke2025potential, sahebi2025dequantizationsupervisedquantummachine, sweke2025kernel}.
    Broadly, a classical surrogate provides a classically efficient parametrization of a function class originally defined by a QML model, without involving any quantum circuit execution at all. A surrogate is considered \emph{accurate} if it is formally guaranteed to contain the original quantum hypothesis family. Since accurate surrogates provide access to the same hypothesis space as the QML model, they are prime candidates for dequantization. This strategy follows the maxim coined in Ref.~\cite{schreiber2023classical}: simply \enquote{train the surrogate}. We term scenarios where this approach succeeds \enquote{dequantization without simulation}.

\subsection{Abstract classification of quantum circuits}\label{ss:classes}

    We propose to organize all possible \emph{parametrized quantum circuits} 
    (PQCs) according to the functions they can give rise to.
    As we detail in Section~\ref{ss:qml}, each PQC gives rise to a set of functions in a natural way.
    We classify the circuits based on properties of their corresponding set of functions.

    \begin{tcolorbox}
    	\begin{class}[Identifiable and evaluatable functions]\label{class:1}
        
    		Circuits that give rise to functions that satisfy both Properties~\ref{pty:evaluatable} and~\ref{pty:identifiable}.
            
    	\end{class} 
    	\begin{class}[Evaluatable functions]\label{class:2}
        
    		Circuits that give rise to functions that satisfy Property~\ref{pty:evaluatable}.
    	\end{class}
    	\begin{class}[Quantum functions]\label{class:3}
        
    		The set of all quantum circuits.
    	\end{class}
    \end{tcolorbox}

    Class~\ref{class:1} includes circuits that can be \emph{classically simulated}.
    Class~\ref{class:2} relates to the \emph{classical surrogates} mentioned above.
    Finally, Class~\ref{class:3} includes general quantum circuits, whose functions may not admit succinct classical representations (according to the standard complexity-theoretic assumption $\BQP\neq\BPP$).
	Notice that the classes are nested, \ie each class is contained in the subsequent ones.
    
\section{Classical representations of PQCs}\label{s:results}

    The way we construct PQCs for learning models paints a rich landscape in which many interesting behaviors can arise.
    In this section, we analyze the general structure of these circuits, specifically highlighting features that can be exploited by classical simulation techniques, which we also review. We later map this structure to the induced hypothesis function families, enabling the classification of QML models with the framework of Section~\ref{ss:classes} from the perspective of quantum advantage in supervised learning tasks.
        
\subsection{Quantum circuit families for machine learning}\label{ss:qml}

	We approach our analysis of circuits used in standard \emph{quantum machine learning} (QML) with the data re-uploading PQCs~\cite{perez2020data} as a unifying architecture. Throughout this analysis, we let $n$ denote the number of qubits in the PQC and restrict our consideration to circuits with total depth at most polynomial in $n$.
	A typical data re-uploading PQC architecture involves:
	\begin{enumerate}
		\item An easy to prepare initial state vector $\lvert\psi_0\rangle$, commonly the computational basis state vector $\ket{0}^{\otimes n}$.
		
		\item A sequence of alternating encoding and trainable layers, $E_{d}({x}),W_d( {\vartheta}_d),E_{d-1}({x}), \dots, W_1( {\vartheta}_1),E_0({x})$, where $d \in \poly(n)$ is the circuit depth. Each $E_{k}({x})$ is a block of data encoding gates dependent on the input data ${x}$, while $W_k( {\vartheta}_k)$ is a block of trainable gates dependent on variational parameters ${\vartheta}_k$. The full parameter set is then ${\vartheta} = ({\vartheta}_1,\dots,{\vartheta}_d)$. 
		
		\item Measurement of a predetermined observable $O$, often chosen to be computationally simple, such as a sum of Pauli operators acting on few qubits (e.g., $\sum_i Z_i$ or single Pauli operators $Z_i$). 
	\end{enumerate}
	In principle, the initial state vector $\lvert\psi_0\rangle$ or the measurement observable $O$ could also depend on input data $x$.
	Unless stated otherwise, we assume they are fixed.
	We extract real-valued outputs from these circuits via expectation values $\langle O \rangle_{{x},{\vartheta}} = \bra{\psi_0}U^\dagger({x},{\vartheta}) O U({x},{\vartheta})\ket{\psi_0}$, where $U({x},{\vartheta})$ is the unitary evolution of the circuit.
	Indeed, these functions $f_\vartheta(x)=\langle O\rangle_{x,\vartheta}$ constitute the hypothesis family $\calF$: 

    \begin{definition}[Hypothesis family of parametrized quantum circuits]\label{def:param_model}
        A parametrized quantum circuit specified by $\lvert\psi_0\rangle$, $U(x,\vartheta)$, and $O$ defines a \emph{parametrized hypothesis family} $\calF$ as:
    	\begin{align}
    		\calF &\coloneqq\{f_\vartheta(x) = \bra{\psi_0}U^\dagger({x},{\vartheta}) O U({x},{\vartheta})\ket{\psi_0}\,|\,\vartheta\in
            \Theta\}.
    	\end{align}
    \end{definition}
    
	\begin{figure}[t]
		\centering
		\includegraphics{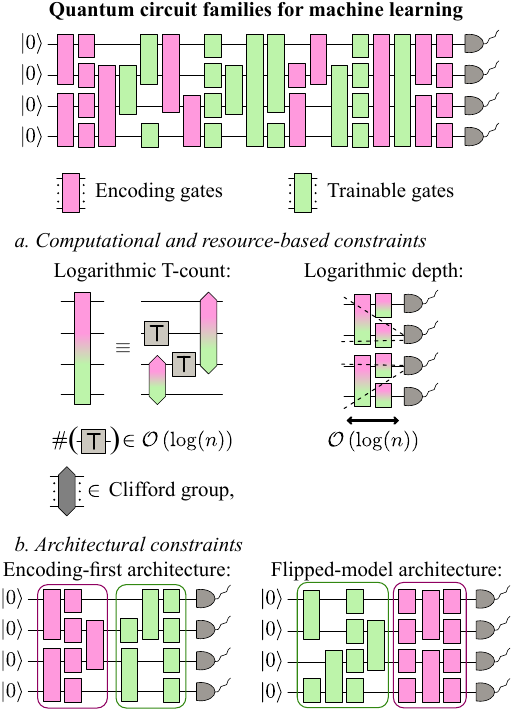}
		\caption{Schematic representation 
        of the structure elements in a PQC. We distinguish encoding gates (pink), parametrized by input data $x$, from trainable gates (light green), parametrized by $\theta$. \textbf{(a) Computational and resource-based constraints:} Blocks may be restricted by a logarithmic number of resource gates. We depict two common examples: magic resources (quantified by $T$-count) and entanglement resources (generated by locally entangling gates). \textbf{(b) Architectural constraints:} The global layout may group blocks of the same type. This work specifically focuses on encoding-first and flipped-model architectures.}
		\label{fig:circuits_intro}
	\end{figure}
	
	These ingredients allow us to classify data re-uploading architectures according to key structural and resource-based properties within this broad class of $n$-qubit, polynomial-depth data re-uploading PQCs. These properties fall into two main categories, as depicted in Fig.~\ref{fig:circuits_intro}.\medskip
	
	\paragraph{Computational and resource-based constraints.} These properties control the computational complexity and expressivity of the PQCs by restricting their non-Clifford content, entanglement depth, or algebraic structure:\medskip
    
	\noindent\textbf{Logarithmic $T$-count in trainable blocks:} The total number $t$ of $T$-gates (or other non-Clifford gates with equivalent cost) across all trainable blocks $W_k( {\vartheta}_k)$ in the circuit fulfills $t\in\calO(\log n)$. All other trainable gates are restricted to the Clifford group. 
    \medskip
	
	\noindent\textbf{Logarithmic $T$-count in encoding blocks:} The total number $t$ of $T$-gates across all data encoding blocks $E_k({x})$ in the circuit fulfills $t\in\calO(\log n)$. All other encoding gates are restricted to the Clifford group. \medskip
	
	\noindent\textbf{Logarithmic-depth encoding layers:} 
	Each encoding block $E_k({x})$ is composed of local gates in a one-dimensional connectivity. The total circuit depth of the data encoding blocks is $\calO(\log n)$. This restricts the amount of entanglement generated across qubits by introducing a light-cone of limited size. \medskip

	\noindent\textbf{Logarithmic-depth trainable layers:} Each trainable block $W_k({\vartheta}_k)$ is composed of local gates in a one-dimensional connectivity. The total circuit depth of the training blocks is $\calO(\log n)$, limiting the circuit's entanglement as for the previous constraint. \medskip
    
    \noindent\textbf{Polynomial-dimensional dynamical Lie group representation:} The dynamical Lie group generated by all gates (both encoding and trainable) used in the PQC, has a representation of dimension in $\poly(n)$. We give more details on the meaning of these terms and their relation to classical simulation in Section~\ref{sss:DLA}.\medskip
	
	\paragraph{Architectural constraints.} These properties, initially highlighted in Ref.~\cite{Jerbi2024shadows}, specify the layout of the encoding and trainable blocks in the circuit:\medskip
	
	\noindent\textbf{Encoding-first architecture:} The PQC structure is restricted to a single data-encoding block $E({x})$ that precedes a single trainable block. The quantum circuit takes the form $U({x},{\vartheta}) = W({\vartheta})E({x})$. \medskip
	
	\noindent\textbf{Flipped architecture:} Similarly, the PQC structure contains a single trainable block $W({\vartheta})$ before a single data encoding block, giving a circuit of the form $U({x},{\vartheta}) = E({x})W({\vartheta})$. \medskip
	
	These two categories reflect complementary aspects of the model: the first (a) constrains the complexity of what the circuit can do, while the second (b) constrains the structure of how it does it.
	
\subsection{Classical simulation of PQCs}\label{ss:simulation}
	
	Efficient classical simulation of quantum circuits (in full or in part) is a crucial tool for the dequantization of the corresponding QML models, as discussed in Section~\ref{ss:dequantization}. Several techniques have been developed to achieve this, each designed to exploit some specific structural aspects of the circuit. Here we survey some of the major approaches.
	
	More precisely, we define \enquote{classical simulation of a PQC} as the classical computation of the respective parametrized functions, defined as expectation values, up to an additive precision $\epsilon$.
    While stronger definitions of simulation exist, we adopt this formulation because it aligns closely with the QML setting discussed above.
    A simulation is deemed \emph{efficient} if the runtime of this computation scales polynomially with both the number of qubits $n$ and the inverse precision $1/\epsilon$. 
	
\subsubsection{Tensor networks}\label{sss:TN}

    Tensor-network methods are a classical simulation technique specialized to quantum systems with low entanglement.
    They are based on representing the amplitudes of the quantum state, as it evolves through the circuit, as a network of contractions between low rank tensors, where the dimension $\chi$ of the connected indices of the tensors is known as the \emph{bond dimension}.
    The most commonly used network architecture is a one-dimensional chain, which gives rise to so-called \emph{matrix product states} (MPS)~\cite{Verstraete_2008}. These have amplitudes in the computational basis $\ket{i_1 , \dots , i_n}$ given by
    \begin{align}
        \PsiMPS^{i_1 ,\dots , i_n} = \sum_{k} (M_1)^{i_1}_{k_1}(M_2)^{i_2 ,k_1}_{k_2} \dots (M_N)^{i_n ,k_{N-1}}.\label{eq:mps_basic}
    \end{align}
    If the entanglement of the evolved state does not grow too much during the circuit, then it is possible to represent the state as an MPS with $\chi\in\poly(n)$, which in turn implies that storing the tensors, updating them and computing expectation values of local observables can be done efficiently. Circuits where the entanglement is guaranteed to remain limited are specifically the ones with a shallow depth, \ie 
    with at most $\calO(\log(n))$ layers of local nearest neighbor gates (in a one-dimensional circuit architecture).
    
    Furthermore, we can use tensor networks also to represent the Heisenberg evolution of observables. In this case the appropriate one-dimensional tensor network structure is known as \emph{matrix product operator} (MPO). As before, an efficient MPO approximation of an evolved operator exists if it has been evolved under a shallow circuit, but also in presence of sufficiently strong circuit noise~\cite{Noh_2020,PhysRevResearch.3.023005}. The relevant quantity here that characterizes the necessary bond dimension is the \emph{local operator entanglement}~\cite{prosen_operator_2007,Bertini2020LOE}. In general, if both a state and an observable can be approximately represented as an MPS (and respectively MPO) with bond dimension $\chi\in\poly(n)$, then the corresponding expectation value can be computed efficiently. It is worth noting that local operator entanglement has been connected to magic \cite{dowling_bridging_2025}, which supports the need for a better understanding of underlying resources.
    
\subsubsection{Stabilizer states}\label{sss:StabStates}

    The stabilizer state formalism~\cite{gottesman1997stabilizercodesquantumerror} can be used to efficiently represent a quantum state evolved by a circuit composed exclusively of Clifford gates~\cite{aaronsonImprovedSimulationStabilizer2004}. If the circuit also contains a number $t$ of non-Clifford gates (for instance $T$-gates), then the resulting state can be represented as a superposition of $\calO(4^t)$ stabilizer states. Manipulating these superpositions and computing expectation values of Pauli observables similarly has a cost of $\calO(4^t)$~\cite{bravyiSimulationQuantumCircuits2019c}. It follows that if $t=\calO(\log(n))$ this type of algorithm is efficient.
    
\subsubsection{Pauli back-propagation}\label{sss:PauliProp}

    Pauli back-propagation methods rely on evolving the observable (which is assumed to be a Pauli operator) in the Heisenberg picture and representing it in terms of its coefficients in the basis of Pauli operators~\cite{schuster2024polynomialtimeclassicalalgorithmnoisy}. As the Pauli basis has $4^n$ elements, this representation is efficient only if the observable is sufficiently sparse, \ie it has support on at most $\poly(n)$ basis elements. This is for sure the case for Clifford circuits supplemented by at most $\calO(\log(n))$ non-Clifford gates (as each of them increases the Pauli basis rank of an observable by a constant factor), \ie in the same setting where stabilizer state methods are efficient.  However, it has been observed that also other settings exist where the evolved observable can be sufficiently well approximated by an expansion in the Pauli basis, suitably truncated to only have $\poly(n)$ terms~\cite{rudolph2023classical,rudolph2025paulipropagation}. This is most clearly the case in the presence of circuit noise \cite{fontana2025classical, martinez2025efficientsimulationparametrizedquantum,schuster2024polynomialtimeclassicalalgorithmnoisy,gonzalez-garcia_pauli_2025}, but has also been shown to be 
    true on average over certain ensembles of random noiseless circuits~\cite{angrisani_classically_2024}.
    
\subsubsection{Dynamical Lie groups and free fermions}\label{sss:DLA}

    The dynamical Lie group is the set of all unitaries that can be obtained by arbitrary products of gates of the form allowed in the circuits we consider. This group acts on the space of operators through its adjoint action $\hat A\mapsto U_g \,\hat A\, U_g^\dagger$. According to representation theory, this implies that the space of operators decomposes into a direct sum of irreducible representation spaces of the group, that is, linear subspaces that are invariant under the group action. Simulation methods based on this (such as $\mathfrak{g}$-sim~\cite{goh2025liealgebraicclassicalsimulationsquantum}) are applicable whenever the observable we intend to measure is contained in such an irreducible representation space of dimension $\poly(n)$. Then, the Heisenberg picture evolution of the observable can be computed efficiently within the low dimensional invariant subspace. This leads to an effective simulation method if, furthermore, the initial state admits efficient expectation values on this subspace of operators.
    
    The most commonly considered setting where these ideas apply is free-fermionic evolution~\cite{terhal_classical_2002} (also known as matchgate circuits~\cite{valiant2001} or fermionic Gaussian unitaries~\cite{guaita2024representationtheorygaussianunitary}). Here the dynamical Lie group is isomorphic to $\mathrm{SO}(2n)$ (where $n$ is the number of fermionic modes) and the $\poly(n)$-dimensional operator subspace is the one spanned by products of up to $k$ Majorana operators, where $k$ is a small constant. Note that it is possibly not a coincidence that fruitful applications of other dynamical Lie groups are very scarce. Indeed, it was recently shown that if one considers gates generated by Pauli operators, the only dynamical Lie group with dimension $\poly(n)$ that can arise is the free-fermionic one~\cite{aguilar2024classificationpauliliealgebras}.
    
\subsubsection{Combinations of approaches}\label{sss:combinations_simulation}
    
    It is important to stress that some of these classical simulation methods can be combined while retaining classical simulation efficiency when applied in the right order. For example,
    the expressions for $\langle O\rangle_{x,\vartheta}$ can still be efficiently computed for state vectors the coefficients of which are given of the form  
    Eq.\ (\ref{eq:mps_basic}),
    capturing MPS, but where the variational set of states is being replaced by
    \begin{equation}
    U(V)
    \sum_{i_1,\dots, i_n}
    \PsiMPS^{i_1 ,\dots , i_n} 
    |i_1,\dots, i_n\rangle,
    \end{equation}
    for $V\in SO(2n)$ reflecting a free-fermionic parity-preserving Majorana mode transformation on $n$ modes and $U(V)$ being its 
    representation in Hilbert space (the metaplectic representation),
    so that $\Theta \simeq 
    \{M_1,\dots, M_N, V\}$. Such approaches are important in classical simulations and dequantizations, as they can combine the capabilities of tensor networks to capture low entanglement states with the possibilities of the free fermionic operations removing entanglement independent of any geometric 
    locality \cite{PhysRevLett.117.210402,krumnow2019overcomingentanglementbarriersimulating,Wu_2025,Schuch_2019}.   
    
    There are further simulation approaches that combine and generalize these settings, specially in the noisy case. For 
    example,
    tensor networks and commuting quantum gate sets with arbitrary geometric locality can be combined \cite{PhysRevB.84.125103}.
    Then, in Ref.~\cite{cristoiu2024}, the author explores a method to simulate settings that are close to $\mathfrak{g}$-sim \cite{goh2025liealgebraicclassicalsimulationsquantum} using pruning techniques that have proven useful for stabilizer states \cite{rudolph2023classical}. Tensor-network methods
    have also been studied in a setting where partial simulation with 
    stabilizer states \cite{PRXQuantum.6.010345,masot2024stabilizer} is possible.
    Ultimately, the success of these methods relies on some equivalent resource $r$, related to those above in a potentially complicated manner, growing as $\calO(\log(n))$. Although we do not discuss the specific structure for each combination, this means at a high level it is easy to translate into the cases we approach, fitting our framework nonetheless.

\subsection{Function families and representations}\label{ss:functions}
	
	The structure of the functions produced by PQCs depends on how the input variables $x$ and trainable parameters $\vartheta$ are encoded into the circuit. This encoding plays a decisive role in determining whether the resulting outputs $f_\vartheta({x})$ can be computed efficiently on a classical computer.
    More specifically, there are different ways in which a classically efficient description of a function may appear. Generally, one may always represent a function by decomposing it on a functional basis as a vector of basis coefficients. From this viewpoint, the circuit implements a map that assigns, to each parameter choice $\vartheta$, a coefficient array that weights a collection of basis functions of the input. A generic decomposition takes the form
	\begin{align}
		f_\vartheta({x}) = \sum_i C_i (\vartheta) T_i(x)\,,
	\end{align}
	where $T_i$ are basis functions (for example Fourier modes or other convenient bases) and $C_i(\vartheta)$ are coefficients that depend on the trainable parameters. In what follows, we always assume that $T_i$ is an appropriately fixed classically efficient functional basis.
	
	Two conceptually different routes then lead to classical efficiency of $f_\vartheta$. The first is sparsity or low-rank: only a polynomial number $m$ of basis elements are needed, in which case efficient description and evaluation are straightforward, provided each $T_i$ can be evaluated efficiently; we refer to this as having low-rank coefficients. The second is compressibility: although the index set of coefficients may be exponentially large, the coefficient array itself factorizes or admits a compact representation (for instance as a tensor network) so that it can be stored and contracted in polynomial time \cite{MANDy,QuantumInspired,B01Goessmann,shin2025newperspectivesquantumkernels}. This can be clearly seen when the basis is local, meaning each basis function factorizes over the input components according to
	\begin{align}\label{eq:function-sparse}
		f_\vartheta({x}) = \sum_{i} C_{i_1,i_2,\dots , i_n}(\vartheta) T_{i_1}(x_1) T_{i_2}(x_2)\dots T_{i_n}(x_n).
	\end{align}
	If the coefficient tensor $C_{i_1,i_2,\dots , i_n}(\vartheta)$ admits an efficient tensor-network decomposition (e.g., it can be represented as an MPS with polynomial bond dimension), then contraction against the local factors yields an output that can be evaluated efficiently for any $x$. In that case, we may write the model as
	\begin{align}\label{eq:MPS-function}
		f_\vartheta({x}) = \sum_{i} \PsiMPS^{i_1, \dots  ,i_n}(\vartheta) T_{i_1}(x_1) T_{i_2}(x_2)\dots T_{i_n}(x_n),
	\end{align}
	emphasizing that the classical cost is controlled by the bond dimension of $\PsiMPS$. 
    
    A related special instance that we may also consider is when the function is sparse in this local product basis. That is 
    \begin{align}
		\sum_{i}^{m} C_{i}(\vartheta) T_{i_1}(x_1) T_{i_2}(x_2)\dots T_{i_n}(x_n)\,,
	\end{align}
	where the number $m$ of multi-index combinations $i_1, \dots, i_n$ that appear in the sum is in $\poly(n)$ to retain an efficient simulation. Note that this is in fact also a subset of tensor-network compressibility: indeed any finite sum of $m$ product terms can be encoded as an MPS with bound-dimension $\chi = m$. As stated above, the two approaches can also be combined, having sparsity and compressibility present in the same ansatz in the above sense.

    \begin{figure}[t]
		\centering
		\includegraphics{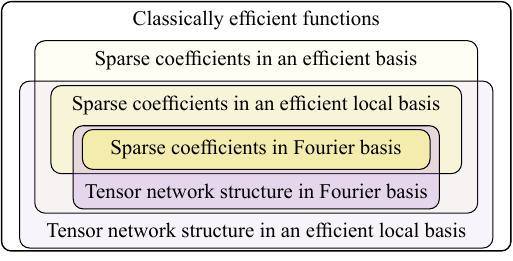}
		\caption{Hierarchy of efficient output function classes for PQCs. The nested structure illustrates the increasing generality of each class, denoted by varying color intensity.
        }    
		\label{fig:func_dec}
	\end{figure}
    
	In practice, a common choice for a local basis is the Fourier basis $T_k(x_j)=e^{ikx_j}$, for which many results exist in the harmonic analysis of PQCs~\cite{schuldEffectDataEncoding2021}; for binary inputs this simplifies to $T_k(x_j)=(-1)^{kx_j}$. 
    For simplicity, we refer to this specific local (product) Fourier basis as simply \enquote{Fourier basis}.
    Figure~\ref{fig:func_dec} summarizes visually the relations between the various cases. It organizes classically efficient functions into nested classes based on whether their efficiency stems from sparsity (few non-zero coefficients) or a compressible structure (an efficient tensor-network representation).
    Notice that the classical simulation techniques described in Section~\ref{ss:simulation} target one or more of the representation mechanisms just described. Tensor-network methods align naturally with tensor-network-structured families; stabilizer and Pauli-sparse methods align with low-rank expansions in the Fourier basis; Lie-algebraic or free-fermion methods correspond to other function spaces of low dimension. Indeed, this classification will allow us to directly connect common PQC ansätze and encoding strategies to known classical representations and classical simulability techniques in the next section.

\subsection{Circuit classification}\label{ss:classification}

	Building on the general classification framework established in Section~\ref{ss:classes}, we now analyze which specific circuit architectures realize each of the defined scenarios. In particular, we examine how combinations of the structural properties introduced in Section~\ref{ss:qml} give rise to the distinct function types discussed in the previous section, thereby falling into one of the classes defined in Section~\ref{ss:classes}. This analysis bridges the abstract and structural perspectives, extending the conceptual overview of Fig.~\ref{fig:func_dec} into a detailed landscape presented in Fig.~\ref{fig:full_picture}.
	
\subsubsection{Shallow depth circuits}\label{sss:shallow-depth}

    First let us consider the case in which some components of the circuit have a shallow, \ie logarithmic, depth. As locality now plays a role, here, tensor-network representations (as introduced in Section~\ref{sss:TN}) are the natural structure that may give rise to classically efficient elements of the model. 
    Notice, however, that such representations will be always efficient only if all components of the circuit are shallow. In these cases the circuits will be fully simulable in their dependence both on $x$ and $\vartheta$. If, on the other hand, only the encoding or only the trainable layers are shallow, then the full circuit will contract in general to an MPS representation with an exponentially large bond dimension. There is not any known way to exploit this partial shallowness. This leads to a first family of circuits which we can attribute to an abstract class:
    
    \begin{observation}[Shallow circuits]
        One-dimensional circuits of local gates, where the total depth of encoding and training layers is at most logarithmic in the system size, are in Class~\ref{class:1}. They can be fully simulated in terms of the training and encoding parameters using MPS methods.
    \end{observation}
    
    As elements of Class~\ref{class:1}, shallow circuits give rise to outputs within a set of classically efficiently evaluatable functions. Indeed, using a construction similar to the one introduced in Ref.~\cite{shinDequantizingQuantumMachine2023} and making some further mild assumptions on the structure of the encoding layers, one can show that the functions $f_\vartheta$ produced by these shallow circuits can themselves be expressed in terms of an efficient MPS structure with respect to a local functional basis, as in Eq.~\eqref{eq:MPS-function}. The coefficients of this representation can furthermore be efficiently identified in terms of the parameters $\vartheta$.
    
    Consider, for instance, data encoding layers composed of a series of single qubit $Z$-rotation gates $S(x)=\bigotimes_j \, \exp(i\frac{\pi}{2} x_j Z_j)$, each acting on a different qubit and depending on an independent data element $x_j$. There may be several such layers separated by other arbitrary entangling gates, provided the total depth of the circuit remains logarithmic. Then one can show that the outcome of measuring a Pauli observable on such circuit can be expressed in the form~\eqref{eq:MPS-function} with respect to the local basis \cite{MANDy,QuantumInspired,B01Goessmann,shin2025newperspectivesquantumkernels}
    \begin{align}
        \{T_k(x)\}_k &
        \coloneqq
        \left\{ T_0(x), \, T_1(x), \, T_2(x) \right\}
        \nonumber
        \\
        &= \left\{ 1, \cos(\pi x), \sin(\pi x) \right\} \,.\label{eq:Tx}
    \end{align}
    The MPS coefficient in this basis can be expressed as a contraction of suitable tensor representations in the Pauli basis of the initial state, observable and the intermediate entangling layers, all of which may depend on $\vartheta$. The shallow depth of the circuit ensures that this contraction is efficient and that the resulting MPS has a polynomial bond dimension.
    See Appendix~\ref{a:shallow} for a more detailed construction. 
        
    \begin{figure*}[t]
		\centering
		\includegraphics{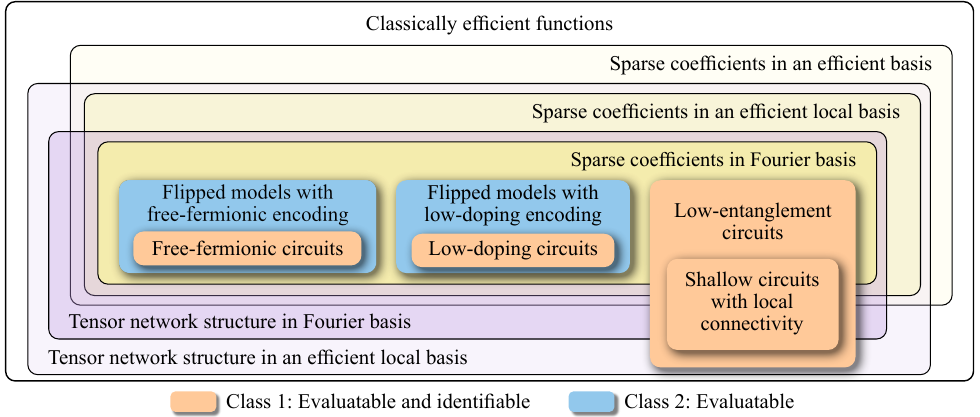}
		\caption{Classification of specific PQC architectures within the functional landscape. Circuit types are mapped onto the background of output functions established in Fig.~\ref{fig:func_dec}, with color coding corresponding to the distinct classes defined in Section~\ref{ss:classes}.
        Note that circuits in Class~\ref{class:3} give rise to classically-inefficient functions, and thus do not fit this Venn diagram.
        }
	\label{fig:full_picture}
	\end{figure*}

\subsubsection{Low doping circuits} \label{sss:low-doping}

    Let us now consider quantum circuits where some components of the circuit have a low $T$-gate count, or more generally can be represented in a sparse way when acting on the Pauli basis. As we have seen in Sections~\ref{sss:StabStates} and~\ref{sss:PauliProp}, this may make the circuits efficiently treatable using classical methods based on stabilizer states or Pauli back-propagation. Certainly, if all encoding and training layers are of this type, then the circuit as a whole can be simulated classically giving rise to classical output functions.
    
    \begin{observation}[Low doping circuits]
        Circuits, where each encoding and training layer, as well as their combined total action, maps (in Heisenberg picture) a Pauli observable to a linear combination of at most $\poly(n)$ Pauli strings, are in Class~\ref{class:1}. They can be fully simulated in terms of the training and encoding parameters using Pauli back-propagation methods.
    \end{observation}
    
    In what follows we will refer to circuits (or parts of circuits) with this property as \emph{low doping} circuits. Notice, however, that this characterization applies to a range of different circuit types. The simplest case is if the circuit in total only contains a number of single qubit non-Clifford gates that is at most logarithmic in the system size. However, also many types of noise lead to a total action of the noisy circuit which can be approximated to a good precision by a sparse outcome on the Pauli basis~\cite{fontana2025classical, martinez2025efficientsimulationparametrizedquantum,rudolph2025paulipropagation}. The same is true for average case circuits sampled from certain ensembles of random noiseless circuits~\cite{rudolph2023classical}. All such circuits fall into this case and give rise to the same type of functions $f_\vartheta$. In this setting, one can also use other stabilizer-related methods~\cite{masot2024stabilizer,PRXQuantum.6.010345, bravyiSimulationQuantumCircuits2019c}.
    
    In general, as the action of the circuit maps the observable to a sum of polynomially many terms, also the outcome $f_\vartheta(x)$ of the circuit will be written as a sum of the polynomially many expectation values of these terms on the initial state. The outcome will, therefore, be an efficient function of the form of Eq.~\eqref{eq:function-sparse} with $m\in\poly(n)$. The precise form of the individual elements of this functional basis will depend on details about the encoding gates, but by assumption it will be efficiently derived from the circuit description as part of the Pauli back-propagation protocol. A reasonable assumption for the encoding layers is again $S(x)=\bigotimes_k \, \exp(i\frac{\pi}{2} x_j Z_j)$, where we now restrict to binary data $x_j\in\{0,1\}$ (otherwise each encoding gate is in general non-Clifford, taking us immediately out of the low-doping scenario). In this case one finds that $f_\vartheta(x)$ is a linear combination of polynomially many terms in the Fourier basis $T_k(x)=\prod_j(-1)^{k_j x_j}$ (see Appendix~\ref{a:low-doping}). 
    
\subsubsection{Partially low doping circuits}\label{sss:part-low-doping}
    
    As before, we may consider what happens if only some parts of the circuit have the low doping property. We argue that the only case in which this can be somehow exploited is if this low-doping component is located at the end of the circuit, just before the final observable. Indeed, if this ``efficient'' component is followed by others with more generic properties, then neither a back-propagated observable nor a forward-propagated state can have a sparse representation on the Pauli basis (or stabilizer state set) which could be exploited.
    
    On the other hand, if the last component of the circuit has a low doping effect, then this can be used to say something about the generic structure of the circuit's outcome. Here we distinguish the case of \emph{flipped} or \emph{encoding-first} architectures, \ie whether the last component is represented by the data-encoding or trainable part of the circuit. 
    
    In the case of flipped architecture circuits, we consider the back-propagated action of the last part of the circuit (\ie the data-encoding block) on the final Pauli observable. If, as we are here assuming, this block has the low doping property introduced above, then by definition the back-propagated observable will have the form 
    \begin{align}
        O(x)
        \coloneqq
        E^\dagger (x) O E(x) = \sum_{k=1}^{\poly(n)} T_k(x) P_k\,,
    \end{align}
    where $P_k$ are Pauli operators. Here, $T_k$ are coefficients whose value depends on the data through the encoding gates in $E(x)$. The specific functional form of $T_k$ will depend on the properties of the encoding gates, but for common choices the coefficients will be efficiently computable classically in terms of $x$ (for a typical choice of encoding gates they will be elements of the discrete Fourier functional basis, as discussed in more detail in Appendix~\ref{a:low-doping}). The precise set of Pauli operators $P_k$ appearing in the sum will depend on the architecture of the encoding circuit, but by construction we assume that it will be a fixed set of polynomially many operators.
    
    If we now consider the value of the final expectation value of the circuit we find
    \begin{align}
        f_\vartheta(x)&=\mathrm{Tr} \rho(\vartheta) O(x) \\
        \nonumber
        &= \sum_{k=1}^{\poly(n)} \mathrm{Tr} \left[\rho(\vartheta)P_k\right] \; T_k(x) \\
        \nonumber
        &= \sum_{k=1}^{\poly(n)} C_k(\vartheta) \; T_k(x) \,,
    \end{align}
    where $\rho(\vartheta)=W(\vartheta)\rho_0 W^\dagger(\vartheta)$.
    Recall that $W(\vartheta)$ denotes the trainable unitary, which is applied initially in such flipped architectures.
    We thus see immediately that $f_\vartheta$ is a classically efficiently evaluatable function according to Property~\ref{pty:evaluatable}, as it is a linear combination of $\poly(n)$ efficient functions. However, it is not necessarily efficiently identifiable (Property~\ref{pty:identifiable}), as the coefficients $C_k(\vartheta)=\mathrm{Tr} \left[\rho(\vartheta)P_k\right]$ cannot be efficiently computed without making further assumptions on the trainable blocks $W(\vartheta)$. 
    
    \begin{observation}[Partially low doping flipped circuits]
        Circuits with a flipped architecture, where the encoding block maps (in Heisenberg picture) a Pauli observable to a linear combination of at most $\poly(n)$ Pauli strings, are in Class~\ref{class:2}. Their outcomes are guaranteed to be efficiently evaluatable functions of $x$, but not necessarily identifiable.
    \end{observation}
    
    Note that in the discussion above we assumed that all the data encoding happens in a circuit block $E(x)$. The same observations would apply also in the slightly more general architecture where one is allowed to measure a generic data-dependent observable $O(x)=\sum_k^{\poly(n)}h_k(x) P_k$, with the assumption that the observable contains a sum of at most polynomially many Pauli terms. 
    
    Finally, the same reasoning can be repeated for encoding-first circuits where the last component of the circuit (\ie the trainable block) has the low doping property. This would give rise to the reversed situation where $f_\vartheta(x)= \sum_{k=1}^{\poly(n)} C_k(x) \; T_k(\vartheta)$. Notice again that with these assumptions $C_k(x)$ are in general not efficient functions, so $f_\vartheta(x)$ of this form does not have the efficient structures of Class~\ref{class:1}, nor of Class~\ref{class:2}.
    
\subsubsection{Free fermionic circuits}\label{sss:free-fermion}

    In the case of polynomial-dimensional dynamical Lie algebras, the analysis is similar to the low doping and the partially low doping cases. This is also the case when the combinations mentioned in Sec.~\label{sss:combinations_simulation} are handled in a setting that ensures some simulability. Regardless, we focus here on the free fermionic case for clarity. If the final component of the circuit is a free fermionic evolution and the observable is a low-degree monomial of Majorana operators, then this can be exploited to find some classically efficient structures in the output $f_\vartheta$.  Indeed, this implies that $f_\vartheta(x)= \sum_{k=1}^{\poly(n)} C_k(x) \; T_k(\vartheta)$ (for encoding-first circuits) or $f_\vartheta(x)= \sum_{k=1}^{\poly(n)} C_k(\vartheta) \; T_k(x)$ (for flipped circuits). As before, this gives rise to evaluatable but not necessarily identifiable functions in the flipped case.
    If also the rest of the circuit is a free fermionic evolution then the whole circuit can be simulated classically.
    
    \begin{observation}[(Partially) free-fermionic circuits]
        Circuits where the total action of encoding and training layers is a free-fermionic evolution and the observable is a low-degree monomial of Majorana operators are in Class~\ref{class:1}. Circuits with a flipped architecture, where the encoding block is a free fermionic evolution, are in Class~\ref{class:2}.
    \end{observation}
    
    As in the previous cases, the form of the functional basis $T_k(x)$ depends on the specific choice of encoding gates. The choice $S(x)=\bigotimes_k \, \exp(i\frac{\pi}{2} x_j Z_j)$ is again a reasonable and common one, as it is fully composed of free fermionic gates. With this choice it follows that $T_k$ is the trigonometric basis discussed in Appendix~\ref{a:shallow}.
    
    Finally note that there are techniques~\cite{miller2025simulationfermionic} that parallel Pauli back-propagation methods for fermions. Similarly to the Pauli case, this implies that there may exist circuits that are to a good approximation free-fermionic due to the presence of noise or of randomness in the parameter distribution.
   
\subsubsection{General circuits}

    We conclude by observing that all other combinations of the quantum circuit structures defined in Section~\ref{ss:qml} produce circuits that satisfy neither Property~\ref{pty:evaluatable} nor Property~\ref{pty:identifiable}.
    They are, therefore, in Class~\ref{class:3}. This includes in particular any circuit architecture which is only partially shallow and only partially low doping (with the exception of the \emph{partially low doping} and \emph{partially free-fermionic} flipped circuits discussed above).
    
    Note that this analysis is based on using the currently known circuit simulation approaches and corresponding functional bases. One cannot rule out that novel simulation techniques may be discovered which provide a systematic way to efficiently represent some circuit architectures in a sparse form on some basis not considered here. Forecasting how these might look like is beyond the scope of this work. More generally, one may consider the problem of finding low-rank representations of an arbitrary state vector or operator~\cite{li2019tutorialcomplexityanalysissingular,Eckart_Young_1936,Sidiropoulos_2024,10.1145/2512329,1384521}. This is, however,  inefficient in the number of qubits and will in any case not necessarily correspond to an efficient low-rank basis for the whole hypothesis function family corresponding to a given PQC. We therefore only focus on what can be achieved provably efficiently with known methods.
    
    Lastly, it is important to notice that there definitely exist circuit architectures where for some specific choices of parameters $\vartheta$ the outcome function $f_\vartheta$ is a classically efficient function. There may even exist fine-tuned architectures where this is true for a range of values of $\vartheta$. For example, rotations with very small angles are known to introduce little non-stabilizerness to the circuit, which would make the circuit simulable~\cite{Bravyi_2019} despite our considerations of structure. Despite this, the argument that we want to make here is that we cannot think of further \emph{generic} circuit architecture that give rise to Class~\ref{class:1} or Class~\ref{class:2} circuits, other than the ones discussed above. Indeed, we find that only an optimization over a sufficiently generic family of circuits is truly in the spirit of ML. Such a large family will contain edge-cases that become classically easy, yet they will not be symptomatic of any explotable behaviours of the rest.
   
\section{Implications on quantum advantage in learning}\label{s:implications}

    A learner is considered to have solved a learning task if it achieves a low expected risk (test error).
    However, the expected risk cannot be minimized directly, since the underlying data distribution is unknown.
    Therefore, a necessary first step toward solving a learning task is minimizing the empirical risk (training error).
    In this section, we first discuss the potential for quantum advantage in \emph{empirical risk minimization} (ERM) within the framework established in Section~\ref{ss:classes}.
    Next, we outline scenarios that may lead to quantum advantage in a broader learning context.
    
\subsection{Empirical risk minimization}\label{ss:imp_ERM}

    To ensure a rigorous comparison, we consider both quantum and classical algorithms for ERM with respect to specific hypothesis families.
    We distinguish scenarios where the quantum learner stems from a circuit family from each of the classes introduced above.
    The circuit family specifies a function basis $(T_k)_k$, which in turn spans a space of linear combinations $f(x)=\sum_k C_k T(x)$.
    We consider ERM over the function family corresponding to these linear combinations.
    Accordingly, for an efficient classical learner to exist, the basis functions $T_k$ must be efficiently evaluatable on a classical computer, and the coefficients $C_k$ must have an efficient structure (\eg sparse or low bond-dimension representations).
    
    The analysis for Class~\ref{class:1} is straightforward.
    Functions arising from circuits in this class are both efficiently evaluatable and identifiable classically.
    Assuming the training algorithm $\calA$ used by the quantum learner is also classically efficient, no quantum advantage for ERM is possible.
    Specifically, for any such quantum learning algorithm, there exists an efficient classical algorithm that accurately replicates every step, reproducing approximately the same output, and hence also its performance. This yields a complete dequantization.

    Class~\ref{class:2} presents a more nuanced case. Here, the functions are classically efficiently evaluatable, but their dependence on the parameters is not necessarily efficiently identifiable.
    We recall that, among the PQC architectures that we classified in Section~\ref{ss:qml}, this class arises in flipped models, \ie the examples presented in Sections~\ref{sss:part-low-doping} and~\ref{sss:free-fermion}. In these cases, the function family defined by the circuit is
    \begin{align}\label{eq:flipped_f}
    	\calF&=\left\{\left.f_\vartheta(x)= \sum_{k=1}^{\poly(n)} \mathrm{Tr} \left[\rho(\vartheta)P_k\right] T_k(x) \,\right|\,\vartheta\in\Theta\right\},
    \end{align}    
    where $T_k$ are classically 
    efficient 
    basis functions 
    and $\rho(\vartheta)=W(\vartheta)\rho_0 W^\dagger(\vartheta)$ is an arbitrarily parametrized state.
    This hypothesis family is a subset of the larger function family: 
    \begin{align}
    	\mathcal{F}'=\left\{\left. f_C(x)= \sum_{k=1}^{\poly(n)} C_k \; T_k(x) \;\right|\; C_k\in\mathbb{R}\right\}\,, \label{eq:F'-convex}
    \end{align} 
    whose functions are also classically efficiently evaluatable.
    To establish a natural comparison, we consider ERM on the larger set $\calF'$ for the classical learner, while the quantum learner is restricted to $\calF$ via the parametrization $\vartheta$.
    The inclusion $\calF\subseteq\calF'$ is clear, although characterizing the exact boundary of $\calF$ is generally difficult~\cite{schuldEffectDataEncoding2021,sweke2025kernel}. This means that the optimal classical solution -- if found -- will achieve an empirical risk that is at least as good as the one of the optimal quantum solution.
    
    Strictly speaking, the quantum and classical learners address slightly different optimization tasks.
    Intuitively, one might expect the classical task to be harder: the classical algorithm optimizes over a larger set $\calF'$  (potentially achieving better performance) and finding the optimal solution could be more difficult due to the increased search space.
    We specialize ERM to the \emph{mean square error} loss $\ell(f(x),y) = (f(x)-y)^2/2$.
    Given a training set $S=\{(x_i,y_i)\}_{i=1}^N$, the empirical risk is $\hR_S[f]=\sum_{i=1}^N(f(x_i)-y_i)^2/(2N)$.
    In this context, the intuition that optimizing over a larger set is harder proves incorrect.
    Indeed, the following proposition shows that a classical learner can exploit the convexity of $\calF'$ to efficiently find the globally-optimal solution.
    
    \begin{proposition}[Informal]\label{prop:quadratic-opt}
        The ERM task with respect to the larger hypothesis family $\calF'$ can be solved efficiently on a classical computer.
    \end{proposition}

    \begin{proof}[Proof sketch:]
    The optimization task can be cast as minimizing a convex quadratic form 
        	\begin{align}
    		\hR_S(f_C) &\propto \sum_{k,k'=1}^{m}  C_k C_{k'} \left(\sum_{i=1}^NT_k(x_i)T_{k'}(x_i)\right) \\
            &\hphantom{\propto} -2 \sum_{k=1}^{m} C_k \left(\sum_{i=1}^Ny_iT_k(x_i)\right) + \sum_{i=1}^Ny_i^2\,
            \nonumber
    	\end{align}
        with respect to the coefficients $C_k$.
    Since the number of basis functions $m$ is polynomial in $n$, the minimization 
    of this quadratic form corresponds to solving a system of linear equations, which can be done classically in polynomial time.

    Alternatively, one may invoke the \emph{representer theorem}~\cite{scholkopf2002learning} from kernel methods. This states that the optimal coefficients take the form $C_k = \sum_{i=1}^N \alpha_i T_k(x_i)$. The optimal coefficients $\alpha$ can be found via \emph{kernel ridge regression}~\cite{scholkopf2002learning}. Here, the kernel function is $K(x,x')=\sum_{k=1}^{\poly(n)} T_k(x)T_k(x')$. Because this involves a polynomial-sized sum of classically efficient functions, the kernel itself is classically efficient to compute. Further details are provided in Appendix~\ref{a:quadratic}.
    \end{proof}

    This observation extends kernel-based surrogates based on \emph{random Fourier features} \cite{landman2022classicallyapproximatingvariationalquantum, sweke2025potential, sahebi2025dequantizationsupervisedquantummachine, sweke2025kernel} in the natural language of our framework.
    In summary, ERM is classically easy whenever the hypothesis family consists of linear combinations of a polynomial-sized set of functions, which includes the only circuit examples in Class~\ref{class:2} that can be constructed with the circuit structure from Sec.~\ref{ss:qml}. Although one can come up with examples where ERM is classically hard, these require circumventing the precise structure of the circuit, falling outside the usual approaches of PQC. In particular, among the contributions of Refs.~\cite{gyurik2023exponential,Jerbi2024shadows} is a learning task where achieving a low empirical risk is expected to be hard for a classical learner, due to a reduction to the Discrete Cube Root (DCR) problem. This case differs from the previous ones in that it uses an efficiently evaluatable hypothesis family that is not of the convex form~\eqref{eq:F'-convex} associated to generic circuit architectures from Sec.~\ref{ss:qml}. This exemplifies a gap between our constructive PQC examples in Class~\ref{class:2} and its general members.
    
    Conversely, we show in the following proposition that ERM restricted to the parametrized quantum hypothesis family $\calF$ is computationally hard in the worst case.
    
    \begin{proposition}[Informal]\label{prop:qcma-complete}
        The ERM task with respect to the restricted hypothesis family $\calF$ cannot be solved efficiently by a quantum computer in the worst case.
    \end{proposition}
    \begin{proof}[Proof sketch:]
        We demonstrate that ERM over the quantum hypothesis family reduces to a QCMA-complete version of the ground-state energy problem.
    
    	Consider the hypothesis class $\mathcal{F}=\{f_\vartheta\,|\,\vartheta\in\Theta\}$ of functions from a PQC with a flipped architecture.
        These functions take 
        the form
        \begin{align}\label{eq:constr_erm}
        	f_\vartheta(x) = \mathrm{Tr}[\rho(\vartheta) O(x)]\,,
        \end{align}
    	where $\rho(\vartheta) = W(\vartheta)\rho_0 W(\vartheta)^\dagger$ is an $n$-qubit quantum state prepared by a polynomial-depth PQC $W(\vartheta)$ specified by classical parameters $\vartheta$, and $O(x)$ is a data-dependent observable.
        
        ERM requires finding $\arg\min_{\vartheta\in\Theta}\hR_S(f_\vartheta)$ with the mean square error loss.
        In Appendix~\ref{a:reduction}, we show this is equivalent to minimizing the energy $\Tr\left[\rho(\vartheta)^{\otimes2}H(S)\right]$,
        where the Hamiltonian $H(S) \propto \sum_{i=1}^N(O(x_i)-y_i\bbI)^{\otimes2}$ encodes the training data.
    	For suitable choices of observables and data, the problem of 
        finding the minimum energy over the restricted set of states $\rho(\vartheta)$ preparable by the PQC is QCMA-complete.
    \end{proof}

    This proposition suggests difficulties for ERM 
    using the hypothesis class of the quantum learning model. While non-convex optimization is generally hard, this worst-case result does not preclude the existence of specific model-task pairs where ERM is quantum efficient.
    Nonetheless, the implication is clear: ERM with Class~\ref{class:2} circuits is worst-case quantum hard.

    Note that Propositions~\ref{prop:quadratic-opt} and~\ref{prop:qcma-complete} present a seemingly paradoxical contradiction: a specific learning task appears classically easy yet quantum hard.
    This would violate the simple principle that classical computers are a subset of quantum computers.
    The subtlety lies here in the definition of the search space.
    By invoking ERM \emph{with respect to different hypothesis families}, we artificially increase the difficulty of the problem for the quantum learner. If the quantum leaner were required to optimize over the larger family $\calF'$, the existence of an efficient classical algorithm would immediately imply also the existence of an efficient quantum algorithm for this.
    Ultimately, it is the non-convexity of the set $\calF$ that makes the problem more complex, independently of the quantum learner's ability  to efficiently \emph{identify} the functions. Thus, without additional structural assumptions, circuit families in Class~\ref{class:2} cannot yield quantum advantage in ERM.

    Finally, we consider Class~\ref{class:3}, which comprises circuits whose hypothesis families are not classically evaluatable (and thus not identifiable).
    Generally, we can rule out the possibility of efficient classical learners solving ERM for this class, as it contains functions that cannot be efficiently evaluated on a classical computer.
    On the other hand, while unstructured quantum circuits are worst-case hard to optimize~\cite{bittel2021training} (consistent with Proposition~\ref{prop:qcma-complete}), known examples of quantum advantage in ERM exist, both in quantum kernel methods~\cite{Liu2021discretelog} and in variational QML~\cite{gilfuster2025relation}.

    For completeness, we note that a small modification to Class~\ref{class:3} would make quantum advantage for ERM impossible.
    If the kernel function $K(x,x')=\sum_k T_k(x)T_k(x')$ were classically efficient to compute, then a similar version of Proposition~\ref{prop:quadratic-opt} would hold, making ERM classically easy again over the larger function family.
    Examples exist of exponentially large function bases that yield efficiently computable kernels, such as \emph{entangled tensor kernels}~\cite{shin2025newperspectivesquantumkernels, sweke2025kernel}.
    Nevertheless, the classical efficiency of the kernel function is an additional assumption that does not hold for Class~\ref{class:3} in general. Therefore, the potential for quantum advantage in ERM remains valid for this class.

\subsection{Learning separations}

    The previous section discussed \emph{empirical risk minimization}  (ERM), namely the first step in supervised learning.
    However, achieving a low empirical risk is an insufficient condition for solving a learning task; the learner must also avoid \emph{overfitting} the training data to ensure \emph{generalization} to unseen data.
    In this section, we discuss scenarios where quantum advantage in learning may exist even when ERM is classically efficient, or where it arises from the improved generalization capabilities of 
    the quantum model.

    We begin with Class~\ref{class:1}, which comprises functions that are both evaluatable and identifiable. For this class, quantum advantage is generally ruled out.
    Since the quantum model is equipped with a classical training algorithm and the resulting function allows for efficient classical evaluation, every step of the learning process can be accurately reproduced on a classical computer. Consequently, the labeling function produced by the quantum learner can be approximated efficiently, placing this quantum learners within the complexity class $\mathsf{P/poly}$ 
    (polynomial-time algorithms with polynomial-size advice).

    Class~\ref{class:2} presents a more nuanced landscape regarding expected risk.
    Recall from Proposition~\ref{prop:quadratic-opt} that a classical learner can efficiently minimize the empirical risk over the relaxed hypothesis family $\calF'$.
    However, success in ERM does not guarantee success in learning: the classical learner might converge to a solution in $\calF'$ that fits the training data perfectly but fails to generalize. Conversely, the quantum learner, restricted to the subset $\calF$, may benefit from an inductive bias that favors generalization.
    
    Recent literature on classical surrogates delineates the boundaries of dequantization for these models~\cite{huang2021power, schreiber2023classical, landman2022classicallyapproximatingvariationalquantum, sweke2025potential, sahebi2025dequantizationsupervisedquantummachine, sweke2025kernel}.
    These works identify sufficient conditions on the task-model pair under which quantum advantage is impossible.
    These results do not claim that dequantization is always possible, thus implicitly stating that there exist specific (potentially contrived and adversarial) task-model pairs outside these conditions where quantum advantage may hold.
    In terms of computational complexity, function families in Class~\ref{class:2} are contained in $\BPPqgenpoly$ (\emph{bounded-error probabilistic polynomial time with polynomial-sized quantum-bounded advice}). Here, the polynomial-size quantum advice overcomes the hardness of identification: the coefficients that specify the function could be inferred from evaluating the function on a polynomial subset of inputs.

    A notable subclass of Class~\ref{class:2} includes so-called \emph{shadow models}~\cite{Jerbi2024shadows} which follow the principle \enquote{train quantum, deploy classical}.
    Examples of learning separations with shadow models often rely on cryptographic hardness assumptions, such as the hardness of discrete cube root problem~\cite{gyurik2023exponential}.
	
	Finally, task-model pairs based on circuit families in Class~\ref{class:3} offer the most straightforward path to quantum advantage.
    Unsurprisingly, the same examples that led to quantum advantage in ERM (such as the discrete logarithm problem~\cite{Liu2021discretelog} or specific variational structures~\cite{gilfuster2025relation}) also yield quantum advantage in the whole learning task.
	The complexity classes typically associated to this idea are $\BQP$ (quantum-efficient functions) and $\Heurpoly$ (classically-learnable functions).
	Formally, achieving quantum advantage requires identifying function families that lie within $\BQP$ but outside $\Ppoly$~\cite{molteni2024exponentialquantumadvantageslearning}.
    
    Recent work~\cite{molteni2025quantummachinelearningadvantages} also explores learning separations using the term \enquote{identification}, yet the context differs significantly from ours. In our framework, identification (Property~\ref{pty:identifiable}) refers to classically-efficiently evaluating $f_{\vartheta}$ given $\vartheta$, a property specific to Class~\ref{class:1}. In contrast, Ref~\cite{molteni2025quantummachinelearningadvantages} defines identification as a specific learning task: recovering the underlying parameters $\vartheta$ of the data-generating process, as opposed to merely learning the function outputs. These frameworks are incomparable, as their task is applicable to functions in any of the three classes we introduce.

\subsection{Approximate randomized dequantization}
	
	Recent developments in dequantization can be interpreted as the effective collapse from Class~\ref{class:3} into Classes~\ref{class:1} and~\ref{class:2} in some settings.
	The study of \emph{barren plateaus}~\cite{larocca2025barren} has notably inspired new approaches in this direction~\cite{bermejo2024quantumconvolutionalneuralnetworks, mele2024noiseinducedshallowcircuitsabsence, miller2025simulationfermionic, angrisani_classically_2024}. The central insight is that, once a PQC architecture has been fixed, even though the function family it gives rise to may be highly complex (not-classically-representable), the highly-complex functions in the class may effectively never appear.
	Formally, consider a parametrized function family $\calF=\{f_\vartheta\,|\,\vartheta\in\Theta\}$ and a probability distribution $\calP(\Theta)$ over the parameters (e.g., the initialization distribution for a gradient-based training algorithm).
	It is possible that, with high probability over $\vartheta\sim\calP$, the realized function $f_\vartheta$ is both efficiently evaluatable and identifiable classically.
	Specifically, one interpretation of the barren plateau phenomenon is that for certain circuit architectures, as the system size increases, the functions sampled from $\calP$ concentrate around the constant function $f_\vartheta(x)=0$ with overwhelming probability. 

    This concentration of measure leads to a trivial \enquote{dequantization} strategy where the classical learner simply ignores the circuit structure and outputs $0$. Despite its simplicity, this strategy successfully replicates the performance of quantum models suffering from barren plateaus with probability approaching $1$.
    
	More sophisticated algorithms, such as those based on (approximate) Pauli-propagation~\cite{rudolph2025paulipropagation,angrisani_classically_2024}, address scenarios where the quantum functions concentrate not around zero, but around non-trivial classically simulatable function families. These cases effectively fall within our Class~\ref{class:1}.
    There also exist intermediate approaches requiring \emph{quantum pre-computation}, i.e.,   extracting some information from the quantum algorithm by performing a few possibly-randomized measurements~\cite{Elben_2022}, which corresponds to Class~\ref{class:2}.
    
	The formal statements of these results take the following form: given the quantum function family $\calF=\{f_\vartheta\,|\,\vartheta\in\Theta\}$, there exists a classical function family $\calF'=\{f'_\vartheta\,|\,\vartheta\in\Theta\}$, such that, with high probability over both inputs $x$ and parameters $\vartheta$, the quantum function $f_\vartheta\in\calF$ is well-approximated by $f'_\vartheta\in\calF'$.
	In terms of complexity classes, this implies that while the full function family may theoretically lie in $\BQP$ and outside $\BPP$, the \emph{effective} family induced by the parameter distribution (i.e., the subset of functions that are likely to be drawn) is well-approximated by a function family in $\BPP$ or $\BPPqgenpoly$.

\section{Conclusions and outlook}\label{s:discussion}

    We have presented a unified framework that synthesizes current approaches to dequantization in variational \emph{quantum machine learning}  (QML).
    Our perspective distinguishes QML models by analyzing the relationship between the \emph{parametrized quantum circuits}  (PQCs) that define them and the classical efficiency of the resulting function families.
    This approach enables a rigorous classification of variational QML models based on two distinct complexity requirements: the efficient \emph{evaluation} of the function outputs and the efficient \emph{identification} of the function from its parameters. Specifically, we identified sufficient conditions for efficient classical simulation (implying both identification and evaluation) and efficient classical surrogation (implying only evaluation). This framework advocates for shifting the study of quantum learning advantage from a purely circuit-based perspective to one focused on the characterization of the function spaces generated by these circuits.

    Applying this classification to the existing literature, we observe that successful dequantization efforts consistently trace back to specific structural building blocks. These blocks reduce the effective complexity of the circuit, linking it to properties such as tensor network factorizability or algebraic low-rankness (e.g., Clifford or free-fermionic structures), which is then reflected on the output family of functions. Our analysis confirms that non-dequantizable PQCs (those capable of exhibiting quantum advantage) usually rely on deep, highly doped circuits that avoid these simplifications. Furthermore, we argue that because efficient representations cannot typically be found via black-box optimization, reliable quantum advantage must be predicated on known structural hardness. In a landscape where many heuristic PQCs have failed to demonstrate advantage, our perspective suggests that future QML designs must explicitly avoid the \enquote{dequantizable building blocks} we have cataloged.

    A central contribution of our perspective is the detailed exploration of Class~\ref{class:2}: the regime of functions that are classically evaluatable but potentially hard to identify. This class simplifies the dequantization landscape by isolating models vulnerable to classical surrogation while retaining quantum features related to parameter hardness. We exemplified this with the case of \emph{partially simulable flipped models} applied to \emph{empirical risk minimization} (ERM). Here, we have demonstrated that no quantum advantage is possible for this task within constructive PQC examples; under these structural constraints, the problem relaxes to a convex quadratic optimization over a fixed, polynomial-sized basis, which acts as an efficient classical surrogate. Conversely, we have highlighted that the native quantum optimization task is, ironically, worst-case hard due to the non-convex constraints of the parametrization.
    For completeness, we also remark an example, outside of the structural PQC assumptions, where quantum advantage is possible in ERM within Class~\ref{class:2}.

    However, Class~\ref{class:2} is not without potential quantum advantage, if one looks beyond ERM. We indeed identified learning models within this class that may exhibit separations, particularly when combined with quantum advice.  Notably, \emph{shadow models} fall into this category; they possess the property of efficient evaluation (allowing classical deployment) but can leverage quantum data access to demonstrate learning separations. Ultimately, our analysis positions Class~\ref{class:3} (circuits generating hypothesis families that generally admit no efficient classical representation) as the primary, though not exclusive, candidate for quantum advantage. 

    Finally, we note that our framework primarily demands \emph{worst-case} simulation hardness for assigning families to Classes~\ref{class:1} or~\ref{class:2}. However, recent research has highlighted the relevance of \emph{average-case} simulation for dequantizing QML in practical scenarios~\cite{cerezo2021variational,bermejo2024quantumconvolutionalneuralnetworks}. Looking forward, we propose bridging these perspectives by analyzing how PQCs in Classes~\ref{class:1} and~\ref{class:2} relate to function families that are simulable in the average case. Investigating whether the \enquote{effective} function families generated by practical training setups collapse from Class~\ref{class:3} to Class~\ref{class:1} % (as seen in the barren plateau induced triviality) 
    represents a critical path toward understanding the realistic power of quantum machine learning.

\subsection*{Acknowledgements}

    The authors thank Vedran Dunjko, Sofiène Jerbi, and Riccardo Molteni for useful comments in a previous version of this manuscript, and Seongwook Shin for insightful discussions. 
    S.M.-L. acknowledges funding from the Spanish Ministry for Digital Transformation and of Civil Service of the Spanish Government through the QUANTUM ENIA project call - Quantum Spain, EU through the Recovery, Transformation and Resilience Plan – NextGenerationEU within the framework of the Digital Spain 2026, as well as the mobility funds from BSC's Severo Ochoa award (CEX2021-001148-S financiado por MCIN/AEI/10.13039/501100011033).
    E.G.-F. is a 2023 Google PhD Fellowship recipient and acknowledges support by the Einstein Foundation (Einstein Research Unit on Quantum Devices), BMBF (Hybrid), and BMWK (EniQmA), as well as travel support from the European Union’s Horizon 2020 research and innovation programme under grant agreement No 951847.
    C.B.-P. also acknowledges support from the BMBF (MUNIQC-Atoms, and Munich Quantum Valley). 
    T.G. acknowledges support from the European Research Council (project DebuQC, grant agreement No. 101098279). 

\bibliography{sources}

\clearpage
\onecolumngrid
\appendix

\setcounter{theorem}{0}
\setcounter{proposition}{0}

\section{Explicit structure of outcome functions}
\subsection{Shallow circuits} \label{a:shallow}

    Consider a circuit of the form discussed in Section~\ref{sss:shallow-depth}, \ie composed of an initial state preparation circuit, some data encoding layers of the form $S(x_{i_1},\dots,x_{i_n})=\bigotimes_k \, \exp(i\frac{\pi}{2} x_{i_k} Z_k)$, possibly further entangling layers between the encoding ones and a final Pauli measurement. The quantum gates in the state preparation circuit, entangling layers and observable may depend on the trainable parameters $\vartheta$ and the total depth of the circuit is constrained to be logarithmic in the qubit number $n$. We will now show how the outcome $f_\vartheta(x)$ of this circuit may be written in the form~\eqref{eq:MPS-function} with an MPS coefficient of polynomial bond dimension. We will use a construction similar to the one introduced in Ref.~\cite{shinDequantizingQuantumMachine2023}.
    First of all, let us represent each element of the circuit in the Pauli basis, defined as the set of operators
    \begin{align}
       P_{\boldsymbol{j}}= \bigotimes_k \; (-i)^{j^0_k j^1_k} \; (Z_k)^{j^0_k} \, (X_k)^{j^1_k} \,,
    \end{align}
    where the 
    multi-index $\boldsymbol{j}\coloneqq (j_0,j_1,\dots,j_n)$ is composed of couples of bits $j_k=(j_k^0,j_k^1)$. 
    As this is a basis of the space of $n$-qubit operators, we can expand on it all the relevant quantities of the problem, including the initial 
    state and final observable
    \begin{align}
         \rho &= \sum_{\boldsymbol{j}} \rho_{\boldsymbol{j}} P_{\boldsymbol{j}}\,, \label{eq:rho-pauli}\\
         O &= \sum_{\boldsymbol{j}} O_{\boldsymbol{j}} P_{\boldsymbol{j}}\,, \label{eq:O-pauli}
    \end{align}
    and all intermediate operators, including the encoding layer $S(x)$ and 
    any entangling layer $W$
    \begin{align}
        S(x)^\dagger P_{\boldsymbol{i}}  S(x) &= \sum_{\boldsymbol{j}} S_{\boldsymbol{i} , \boldsymbol{j}}(x) P_{\boldsymbol{j}} \,,\label{eq:S-pauli}\\
        W^\dag P_{\boldsymbol{i}}  W &= \sum_{\boldsymbol{j}} W_{\boldsymbol{i} ,\boldsymbol{j}} P_{\boldsymbol{j}} \,.
    \end{align}
    
    The shallow depth nature of the circuit elements implies that  each of the tensor objects $\rho_{\boldsymbol{j}}$, $O_{\boldsymbol{j}}$, $S_{\boldsymbol{i},  \boldsymbol{j}}(x)$ and $W_{\boldsymbol{i},  \boldsymbol{j}}$ can be expressed as a $\poly(n)$ bond-dimension MPS or MPO by standard tensor network techniques. If these circuit elements are constructed from gates that depend explicitly on the trainable parameters $\vartheta$, then also the tensors of the corresponding MPS/MPO representation will depend efficiently on these parameters. In particular, some straightforward calculations (which can be checked by substitution) allows us to express $S_{\boldsymbol{i},  \boldsymbol{j}}(x)$ as
    \begin{align}
        S_{\boldsymbol{i},  \boldsymbol{j}}(x)=\sum_{\alpha} \, S^{\alpha}_{\boldsymbol{i}, \boldsymbol{j}} \:  T^{(1)}_{\alpha_1}(x_1)\cdots T^{(n)}_{\alpha_n}(x_n)=\sum_{\alpha} \prod_k \, S^{\alpha_k}_{i_k , j_k} \:  T^{(k)}_{\alpha_k}(x_k) \: , \label{eq:S-tensor}
    \end{align}
    where $\alpha=(\alpha_1,\dots,\alpha_n)\in\{0,1,2\}^n$, $k\in\{1,\dots ,n\}$ and for any $k$ $T^{(k)}(x) \coloneqq \big\{ T_0(x), \, T_1(x), \, T_2(x) \big\} = \big\{ 1, \cos(x_k \pi), \sin(x_k \pi) \big\}$, with the indices $\alpha_k$ fixing $T^{(k)}_{\alpha_k}(x_k)$ to one of these functions. Lastly, $S^{\alpha_k}_{i_k , j_k}$ is given by
    \begin{align}
        S^0_{i_k , j_k} &= \delta_{i^1_k,0} \: \delta_{i^1_k,j^1_k} \: \delta_{i^0_k,j^0_k},\\
        S^1_{i_k , j_k} &= \delta_{i^1_k,1} \: \delta_{i^1_k,j^1_k}\: \delta_{i^0_k,j^0_k}, \\
        S^2_{i_k , j_k} &= \delta_{i^1_k,1}  \: \delta_{i^1_k,j^1_k} \: \delta_{i^0_k+j^0_k, 1} \:(-1)^{j_k^0} \,.\label{eq:S_alpha}
    \end{align}
    Consider now the case of a simple circuit with a single encoding layer $S(x)$. We assume all other gates are included in the definition of a (possibly $\vartheta$-dependent) initial state $\rho$ and observable $O$. In this case, the circuit outcome is given by $f_\vartheta(x)=\Tr{\rho S(x)^\dag O S(x)}$, which combining equations~\eqref{eq:rho-pauli},~\eqref{eq:O-pauli},~\eqref{eq:S-pauli} and~\eqref{eq:S-tensor} above can be expressed as
    \begin{align}
        f_\vartheta(x) =\sum_\alpha\, \left[\sum_{\boldsymbol{i},\boldsymbol{j}} \rho_{\boldsymbol{i}} S^{\alpha}_{\boldsymbol{i}, \boldsymbol{j}} O_{\boldsymbol{j}}\right]  \;T^{(1)}_{\alpha_1}(x_1)\cdots T^{(n)}_{\alpha_n}(x_n)\,. \label{eq:MPS-result}
    \end{align}
    We recognize a function of the efficient form~\eqref{eq:MPS-function}, where the MPS coefficient is given by the expression in the square bracket, which can be seen as a tensor contraction of the MPS/MPO objects corresponding to  $\rho_{\boldsymbol{i}}$, $S^{\alpha}_{\boldsymbol{i}, \boldsymbol{j}}$ and $O_{\boldsymbol{j}}$, as illustrated in Figure~\ref{fig:MPS}. As is clear from the diagram, contracting this will produce an MPS coefficient with bond dimension equal to the product of the bond dimensions of $\rho_{\boldsymbol{i}}$ and $O_{\boldsymbol{j}}$, which is polynomial by assumption. 

    \begin{figure}[ht]
        \centering
        \includegraphics[width=0.25\linewidth]{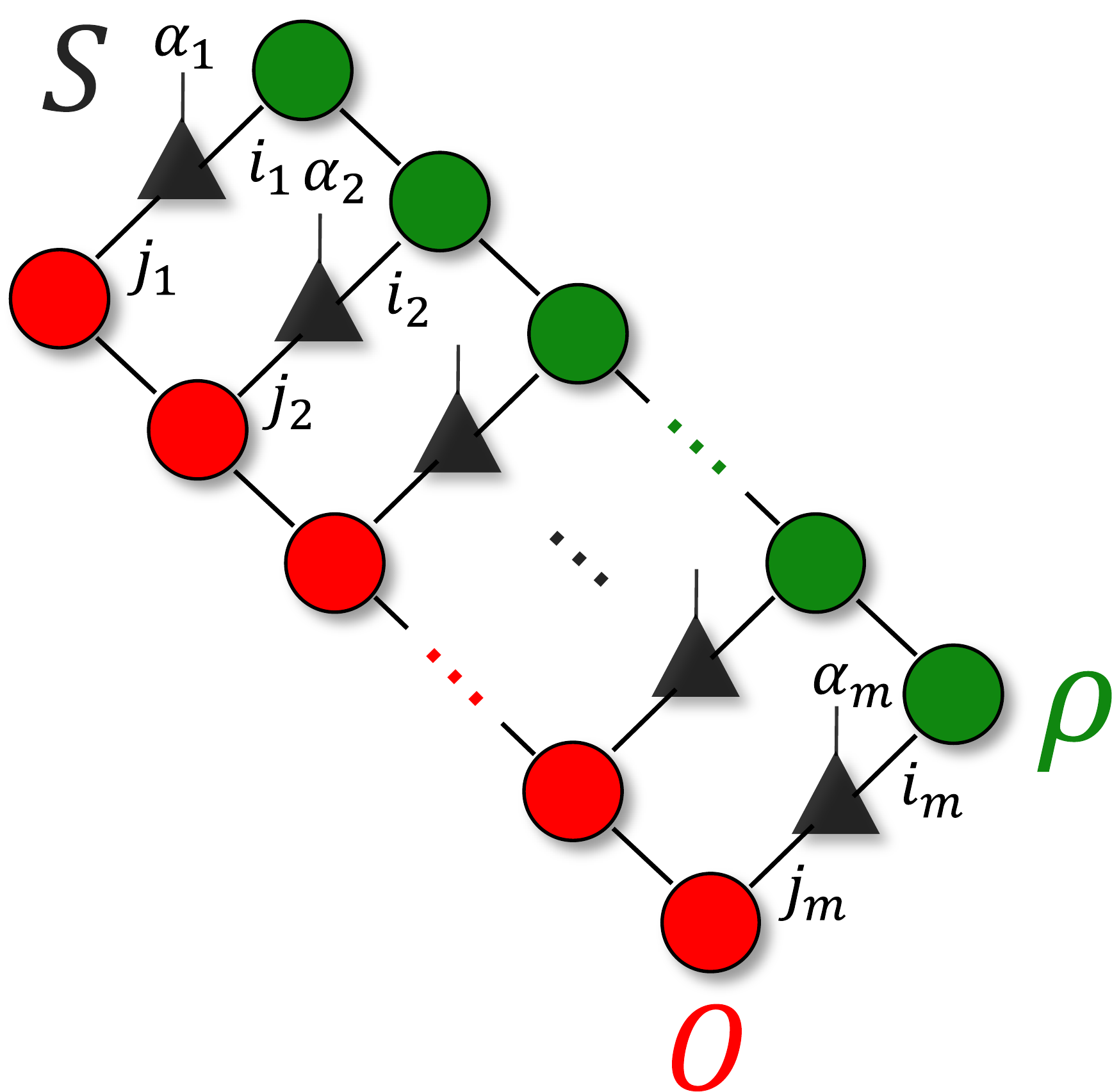}
        \caption{A tensor diagram representation of the expression~\eqref{eq:MPS-result} in terms of the MPS/MPOs for $\rho_{\boldsymbol{i}}$, $S^{\alpha}_{\boldsymbol{i}, \boldsymbol{j}}$ and $O_{\boldsymbol{j}}$.}
        \label{fig:MPS}
    \end{figure}

    Let us then address the more general situation where one has several encoding layers $S(x)$ separated by entangling circuit blocks $W$. As we are considering circuits of logarithmic depth, also each individual block $W$ should be representable as an MPO of polynomial bond dimension so the combined bond dimension of all these $W$ blocks also remains polynomial. Combining these elements together as before we find that the total output of the circuit can be expressed as
    \begin{align}
        f_\vartheta(x)=\sum_{\alpha} C_{\alpha_1, \dots, \alpha_m} T^{(1)}_{\alpha_1}(x_1)\cdots T^{(m)}_{\alpha_m}(x_m)\,, \label{eq:PEPS-result}
    \end{align}
    for a circuit with $m$ encoding gates. Here the coefficient tensor $C_{\alpha_1, \dots, \alpha_m}$ has the two-dimensional structure illustrated in Figure~\ref{fig:PEPS}. This object can be recast into a one-dimensional MPS structure by choosing an ordering of the legs $\alpha_i$. If this is done in the manner illustrated in the figure (\ie `snaking' first along the shallow depth of the circuit and then along its width), the resulting MPS will have a polynomial bond dimension, as each bipartition of the indices $\alpha_i$ only cuts a logarithmic number of bonds of the total tensor.

    \begin{figure}[hb]
        \centering
        \includegraphics[width=0.5\linewidth]{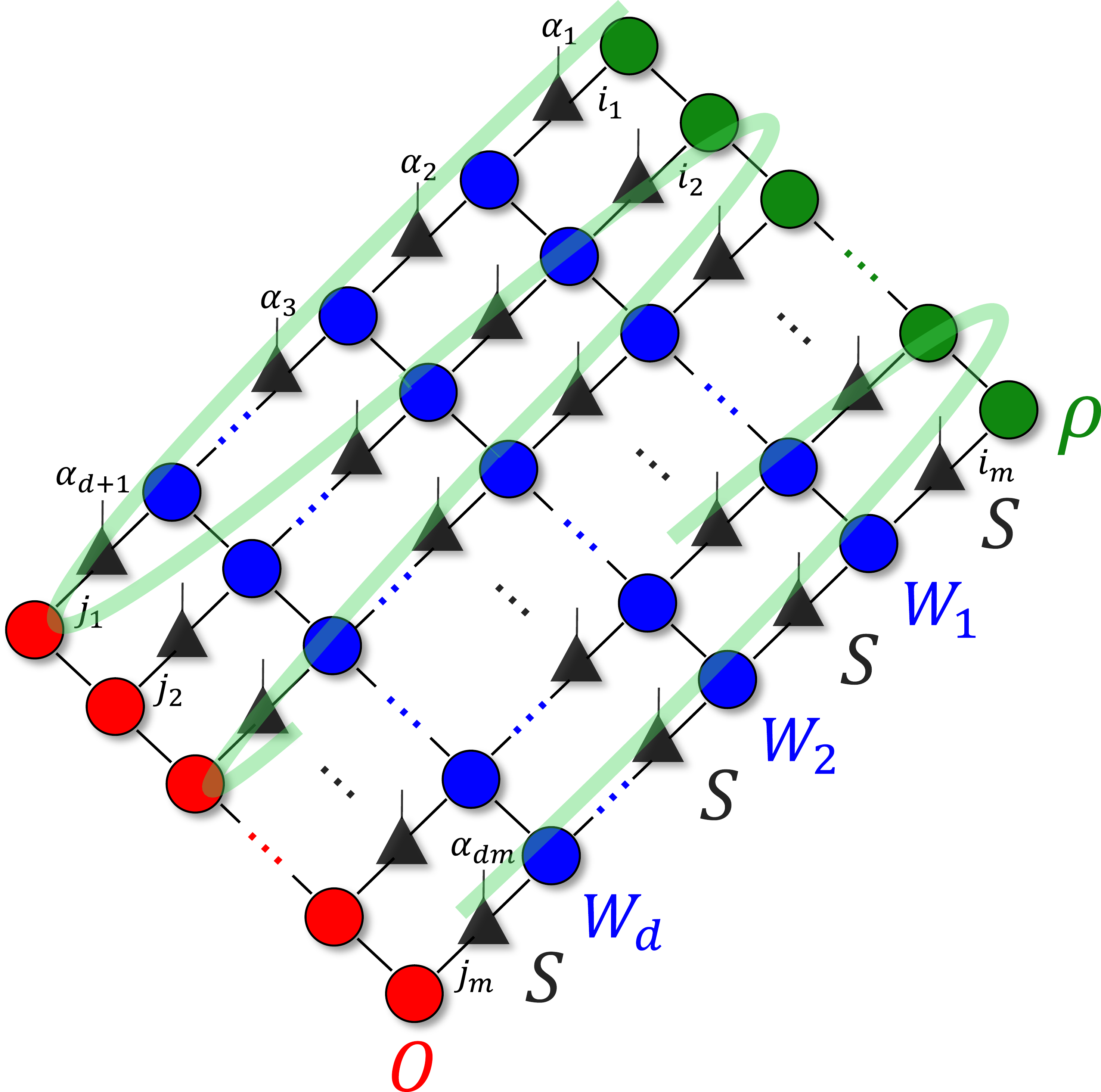}
        \caption{A tensor diagram representation of the expression~\eqref{eq:PEPS-result} in terms of the MPS/MPOs for $\rho_{\boldsymbol{i}}$, $S^{\alpha}_{\boldsymbol{i}, \boldsymbol{j}}$, $W^{\alpha}_{\boldsymbol{i}, \boldsymbol{j}}$ and $O_{\boldsymbol{j}}$.}
        \label{fig:PEPS}
    \end{figure}

    Although the same manipulation can be done for circuits without assumptions on the depth, in that case we cannot ensure the final MPS is efficient, and in general it will not be.
    
    \subsection{Low doping circuits}\label{a:low-doping}
    
    We consider circuits of the low doping form discussed in Section~\ref{sss:low-doping}. As in the previous Section~\ref{a:shallow}, we represent the circuits in the Pauli basis. As we now assume binary inputs $x_j\in\{0,1\}$, the matrix $S_{\boldsymbol{i},  \boldsymbol{j}}(x)$ representing the encoding layer $S(x)=\bigotimes_k \, \exp(i\frac{\pi}{2} x_{i_k} Z_k)$ reduces to the form 
    \begin{align}
        S_{\boldsymbol{i},  \boldsymbol{j}}(x)=\sum_{\alpha} \prod_k \, S^{\alpha_k}_{i_k , j_k} \:  T_{\alpha}(x) \label{eq:S-tensor-low}
    \end{align}
    where $\alpha=(\alpha_1,\dots,\alpha_n)\in\{0,1\}^n$ and $T_\alpha(x) = (-1)^{\alpha \cdot x}\coloneqq\prod_k (-1)^{\alpha_k x_k}$ and $S^{\alpha_k}_{i_k , j_k}$ is given by
    \begin{align}
        S^{\alpha_k}_{i_k , j_k} = \delta_{i^1_k,\alpha_k} \: \delta_{i^1_k,j^1_k} \: \delta_{i^0_k,j^0_k} \,.
    \end{align}
    Considering a generic circuit composed of several encoding layers separated by arbitrary unitaries $W^{(k)}$, possibly depending on $\vartheta$, the final outcome is given by
    \begin{align}
        f_\vartheta(x)&=\sum_{\substack{\boldsymbol{i}_1,\dots, \boldsymbol{i}_{m+1}\\\boldsymbol{j}_1,\dots, \boldsymbol{j}_{m} }} \rho_{\boldsymbol{i}_1} S^{(1)}_{\boldsymbol{i}_1,\boldsymbol{j}_1}(x) W_{\boldsymbol{j}_1,\boldsymbol{i}_2}^{(1)} \cdots  S^{(m)}_{\boldsymbol{i}_{m-1},\boldsymbol{j}_m}(x)W_{\boldsymbol{j}_m,\boldsymbol{i}_{m+1}}^{(m)} O_{\boldsymbol{i}_{m+1}} \\
        &=\sum_{\boldsymbol{i}_1,\dots, \boldsymbol{i}_{m+1}} \rho_{\boldsymbol{i}_1} W_{\boldsymbol{i}_1,\boldsymbol{i}_2}^{(1)} \cdots  W_{\boldsymbol{i}_m,\boldsymbol{i}_{m+1}}^{(m)} O_{\boldsymbol{i}_{m+1}} \: \prod_{k=1}^m (-1)^{\boldsymbol{i}_k\cdot x_k}\,.
    \end{align}
    By assumption, each individual circuit layer maps a Pauli basis element to at most $\poly(n)$ other basis elements, and so does the total action of the circuit (which in particular means there can be at most logarithmically many such layers). This implies that in the sum above only at most $\poly(n)$ combinations of indices $\boldsymbol{i}_1,\dots, \boldsymbol{i}_{m+1}$ can correspond to non-vanishing terms. We thus see that $f_\vartheta$ is a polynomial sum of elements of the Fourier functional basis.

\section{Proof of Proposition 1}\label{a:quadratic}
    Here we re-state and prove Proposition~\ref{prop:quadratic-opt}.

    \begin{proposition}[Rigorously restated]
        Let the supervised learning task be defined by a training set $S = \{(x_i, y_i)\}_{i=1}^{N}$. The classical task involves the hypothesis class $\mathcal{F}'$ of functions $f_C : \mathcal{X} \to \mathbb{R}$ defined as linear combinations of $m = \poly(n)$ basis functions $T_k(x)$ as
        \begin{align}
            f_C(x) = \sum_{k=1}^{m} C_k T_k(x)\,,
        \end{align}
        where $\mathbf{C} = (C_1, \dots, C_m) \in \mathbb{R}^m$ is a vector of real-valued coefficients. We assume each basis function $T_k(x)$ is efficiently computable.
        The learning objective is to find the optimal coefficient vector $\mathbf{C}^\ast$ that minimizes the empirical risk, defined by the mean squared error
        \begin{align}
            R_S[f_C] = \frac{1}{N} \sum_{i=1}^{N} (f_C(x_i) - y_i)^2\,.
        \end{align}
        The problem of finding the coefficient vector $\mathbf{C}^\ast$ that minimizes the empirical risk $R_S[f_C]$ is equivalent to minimizing the  degree-2 polynomial 
        \begin{align}
            Q(\mathbf{C}) = \sum_{k=1}^{m} \sum_{k'=1}^{m} C_k C_{k'} M_{k,k'} - 2 \sum_{k=1}^{m} C_k V_k + Z\,
        \end{align}
        (a convex quadratic form) in the coefficients $C_k$, where
        \begin{itemize}
            \item $M_{k,k'} = \frac{1}{N} \sum_{i=1}^{N} T_k(x_i) T_{k'}(x_i)$,
            \item $V_k = \frac{1}{N} \sum_{i=1}^{N} y_i T_k(x_i)$,
            \item $Z = \frac{1}{N} \sum_{i=1}^{N} y_i^2$.
        \end{itemize}
        The minimization of this quadratic form can be solved classically in polynomial time.
    \end{proposition}
    \begin{proof}
        The proof proceeds by direct algebraic expansion of the empirical risk functional $R_S[f_C]$. First, we substitute the definition of the classical model $f_C(x)$ into the risk functional
        \begin{align}
            R_S[f_C] = \frac{1}{N} \sum_{i=1}^{N} \left( \left(\sum_{k=1}^{m} C_k T_k(x_i)\right) - y_i \right)^2\,.
        \end{align}
        We expand the square term and distribute the outer summation
        \begin{align}
            R_S[f_C] &= \frac{1}{N} \sum_{i=1}^{N} \left[ \left(\sum_{k=1}^{m} C_k T_k(x_i)\right)^2 - 2y_i\left(\sum_{k=1}^{m} C_k T_k(x_i)\right) + y_i^2 \right] \\
            \nonumber
            &= \frac{1}{N} \left[ \sum_{i=1}^{N} \left(\sum_{k=1}^{m} C_k T_k(x_i)\right) \left(\sum_{k'=1}^{m} C_{k'} T_{k'}(x_i)\right) - \sum_{i=1}^{N} 2y_i\left(\sum_{k=1}^{m} C_k T_k(x_i)\right) + \sum_{i=1}^{N} y_i^2 \right]\,.
        \end{align}
        Now, we reorder the summations for the quadratic and linear terms in $\mathbf{C}$
        \begin{align}
            R_S[f_C] &= \frac{1}{N} \left[ \sum_{k=1}^{m}\sum_{k'=1}^{m} C_k C_{k'} \left(\sum_{i=1}^{N} T_k(x_i) T_{k'}(x_i)\right) - 2\sum_{k=1}^{m} C_k \left(\sum_{i=1}^{N} y_i T_k(x_i)\right) + \sum_{i=1}^{N} y_i^2 \right] \\
            &= \sum_{k,k'} C_k C_{k'} \underbrace{\left( \frac{1}{N}\sum_{i=1}^{N} T_k(x_i) T_{k'}(x_i) \right)}_{M_{k,k'}} - 2\sum_{k} C_k \underbrace{\left( \frac{1}{N}\sum_{i=1}^{N} y_i T_k(x_i) \right)}_{V_k} + \underbrace{\left( \frac{1}{N}\sum_{i=1}^{N} y_i^2 \right)}_{Z}\,.
             \nonumber
        \end{align}
        This is exactly the quadratic form $Q(\mathbf{C}) = \sum_{k,k'} C_k C_{k'} M_{k,k'} - 2\sum_k C_k V_k + Z$, which proves the equivalence between terms. Lastly, the Hessian of such a problem is the matrix $M$ with entries $M_{i,j} =\sum_k T_k(x_i)T_k(x_j)$, which is a kernel matrix~\cite{scholkopf2002learning}. Because kernel problems are \emph{positive semi-definite} (PSD), the Hessian will 
        also be PSD, hence it is a convex problem and efficiently solvable.
    \end{proof}

\section{Proof of Proposition 2}\label{a:reduction}
    
    Here we re-state and prove Proposition~\ref{prop:qcma-complete}.

    \begin{proposition}[Rigorously restated]
        Let the quantum task be defined by a hypothesis class $\mathcal{F}$ of functions $f_\vartheta: \mathcal{X} \to \mathbb{R}$ from a flipped architecture PQC. These functions are of the form
        \begin{align}
        f_\vartheta(x) = \mathrm{Tr}[\rho(\vartheta) O(x)]\,,
        \end{align}
        where $\rho(\vartheta) = W(\vartheta)\rho_0 W(\vartheta)^\dagger$ is an $n$-qubit quantum state prepared by a PQC $W(\vartheta)$ specified by classical parameters $\vartheta$, and $O(x)$ is a data-dependent observable. The learning objective is to find the optimal parameters $\vartheta^*$ that minimize the empirical risk $R_S[f_\vartheta] = \frac{1}{N} \sum_{i=1}^{N} (f_\vartheta(x_i) - y_i)^2$ over the training set $S = \{(x_i, y_i)\}_{i=1}^{N}$.
        The quantum learning task of finding $\min_{\vartheta} R_S[f_\vartheta]$ is equivalent to finding the state $\rho(\vartheta^*)$ that minimizes the energy of a corresponding $2n$-qubit Hamiltonian $H(S)$
        \begin{align}
        \min_{\vartheta} R_S[f_\vartheta]
        \coloneqq 
        \min_{\vartheta} \mathrm{Tr}\left[\left(\rho(\vartheta) \otimes \rho(\vartheta)\right) H(S)\right]\,,
        \end{align}
        where the Hamiltonian $H(S)$ is constructed from the training data as
        \begin{align}
        H(S) = \frac{1}{N} \sum_{i=1}^{N} \left( O(x_i) - y_i \mathbb{I} \right) \otimes \left( O(x_i) - y_i \mathbb{I} \right)\,.
        \end{align}
        The problem of finding the minimum energy over the restricted set of states $\rho(\vartheta)$ preparable by the PQC is QCMA-complete.
    \end{proposition}
    \begin{proof}
        The proof consists of two parts. First, we show that the problem is in the complexity class QCMA. Second, we show that it is QCMA-hard by reducing a known QCMA-complete problem to it.

        A problem is in QCMA if a ``yes" instance can be verified in polynomial time by a quantum computer given a classical string as a proof (witness).
\begin{itemize}
    \item \textbf{Classical witness:} The witness for a YES instance is the classical string describing the parameters $\vartheta$ that allegedly achieve a low empirical risk, i.e., $R_S[f_\vartheta] \le a$.
    \item \textbf{Quantum verifier:} The verifier is a quantum circuit that takes the witness $\vartheta$ and the training set $S$ as input. To verify the claim, the verifier must estimate the empirical risk $R_S[f_\vartheta] = \frac{1}{N}\sum_{i=1}^N (f_\vartheta(x_i) - y_i)^2$. For each data point $(x_i, y_i) \in S$, the verifier can efficiently estimate the term $f_\vartheta(x_i) = \text{Tr}[\rho(\vartheta)O(x_i)]$ by the following steps.
    \begin{enumerate}
        \item Preparing the state $\rho(\vartheta) = W(\vartheta)\rho_0 W(\vartheta)^\dagger$ by applying the circuit $W(\vartheta)$ (constructed from the classical witness $\vartheta$) to a standard initial state $\rho_0$.
        \item Measuring the observable $O(x_i)$ to estimate its expectation value. This can be done efficiently if $O(x_i)$ is a sum of a polynomial number of Pauli strings.
    \end{enumerate}
    The verifier repeats this process to obtain a sufficiently accurate estimate of each $f_\vartheta(x_i)$, computes the total empirical risk, and accepts if the estimate is $\le a$ (or, more precisely, below a threshold between $a$ and $b$). The promise gap $b-a \ge 1/\text{poly}(n)$ ensures this verification is reliable. Since the witness is classical and the verifier is an efficient quantum circuit, the problem is in QCMA.
\end{itemize}

The second part of the proof consists of a reduction of the risk minimization to an energy minimization, and then illustrating the computational hardness of the resulting problem. We start by stating the empirical risk functional and rewrite it algebraically, to get
        \begin{align}
            R_S[f_\vartheta] = \frac{1}{N} \sum_{i=1}^{N} (\mathrm{Tr}[\rho(\vartheta) O(x_i)] - y_i)^2\,.
        \end{align}
        Since $\mathrm{Tr}[\rho(\vartheta)] = 1$, we can express the scalar $y_i$ as $\mathrm{Tr}[\rho(\vartheta) (y_i \mathbb{I})]$. This allows us to combine the terms inside the trace as
        \begin{align}
            \mathrm{Tr}[\rho(\vartheta) O(x_i)] - y_i = \mathrm{Tr}[\rho(\vartheta) (O(x_i) - y_i \mathbb{I})]\,.
        \end{align}
        Let $A_i = O(x_i) - y_i \mathbb{I}$. We use the identity $(\mathrm{Tr}[\rho A])^2 = \mathrm{Tr}[(\rho \otimes \rho) (A \otimes A)]$, in order to get
        \begin{align}
            (\mathrm{Tr}[\rho(\vartheta) A_i])^2 = \mathrm{Tr}[ (\rho(\vartheta) \otimes \rho(\vartheta)) (A_i \otimes A_i) ].
        \end{align}
    Substituting this back into the expression for the risk, we obtain
    \begin{align}
            R_S[f_\vartheta] &= \frac{1}{N} \sum_{i=1}^{N} \mathrm{Tr}\left[ (\rho(\vartheta) \otimes \rho(\vartheta)) (A_i \otimes A_i) \right] \\
             \nonumber
            &= \mathrm{Tr}\left[ (\rho(\vartheta) \otimes \rho(\vartheta)) \left( \frac{1}{N} \sum_{i=1}^{N} (A_i \otimes A_i) \right) \right] \\
             \nonumber
            &= \mathrm{Tr}\left[ (\rho(\vartheta) \otimes \rho(\vartheta)) H(S) \right]\,,
        \end{align}
        which completes the reduction. We now prove that the problem of minimizing this energy is QCMA-hard. We do this by reducing the known QCMA-complete problem, \textit{low complexity low energy states}~\cite{wocjan2003two}, to our problem. An instance of this problem consists of a $3$-local Hamiltonian $H'$, a complexity bound $k$ (number of elementary gates), and two real numbers $\alpha, \beta$ with a promise gap $\beta - \alpha \ge 1/\text{poly}(n)$. The problem is to decide if:
    \begin{enumerate}
        \item (YES case) There exists a quantum circuit $V$ of at most $k$ gates such that the state vector $|\psi\rangle = V|0\rangle^{\otimes n}$ has energy $\langle\psi|H'|\psi\rangle \le \alpha$.
        \item (NO case) All quantum circuits $V$ of at most $k$ gates prepare state vectors $|\psi\rangle$ with energy $\langle\psi|H'|\psi\rangle \ge \beta$.
    \end{enumerate}
    We construct an instance of our learning problem from this QCMA-complete instance as follows:
    \begin{itemize}
        \item \textbf{PQC architecture:} The family of circuits $W(\vartheta)$ is defined as the set of all quantum circuits with at most $k$ gates. The classical parameters $\vartheta$ are a description of such a circuit.
        \item \textbf{Training set:} We use a single data point, $S = \{(x_1, y_1)\}$.
        \item \textbf{Observable and label:} We set the observable to be the given Hamiltonian, $O(x_1) = H'$, and the label to be zero, $y_1 = 0$.
    \end{itemize}
    With this setup, the Hamiltonian for our learning task, $H(S)$, becomes
    \begin{align}
        H(S) = (O(x_1) - y_1 \mathbb{I}) \otimes (O(x_1) - y_1 \mathbb{I}) = H' \otimes H'\,.
    \end{align}
    The learning task is to find $\min_{\vartheta} R_S[f_\vartheta]$. As we have seen previously, the empirical risk for a quantum state $\rho(\vartheta) = W(\vartheta)\rho_0W^\dagger(\vartheta)$ becomes the squared expectation value 
        \begin{align}
            R_S[f_\vartheta] = \mathrm{Tr}\left[\left(\rho(\vartheta) \otimes \rho(\vartheta)\right) (H' \otimes H') \right] = \left(\mathrm{Tr}[\rho(\vartheta) H']\right)^2\, 
        \end{align}
        of this Hamiltonian. Minimizing the risk over all PQC parameters $\vartheta$ is therefore equivalent to finding the minimum of $(\langle\psi|H'|\psi\rangle)^2$ over the set of all state vectors $|\psi\rangle$ preparable by circuits with at most $k$ gates. This directly maps the promise gap of the 
        \emph{low complexity low energy states problem}:
        
    \begin{itemize}
        \item \textbf{YES case:} If there is a preparable state vector $|\psi\rangle$ 
        with energy $\langle\psi|H'|\psi\rangle \le \alpha$, then the minimum risk is $\le \alpha^2$.
        \item \textbf{NO case:} If all preparable state vectors $|\psi\rangle$ have energy $\langle\psi|H'|\psi\rangle \ge \beta$, then the minimum risk is $\ge \beta^2$.
    \end{itemize}
    This provides a valid reduction, mapping the QCMA-complete problem to our learning problem with risk thresholds $a=\alpha^2$ and $b=\beta^2$. The promise gap remains inverse-polynomial, so the problem is QCMA-hard.   
    Since the problem is in QCMA and is QCMA-hard, it is QCMA-complete.
    \end{proof}

\end{document}